\documentclass[10pt,journal,twoside,web]{ieeecolor}
\usepackage{generic}
\usepackage{cite}
\usepackage{amsmath,amssymb,amsfonts}
\usepackage{algorithmic}
\usepackage{graphicx}
\usepackage{textcomp}
\usepackage[export]{adjustbox}
\usepackage{balance} % for last page
\usepackage{stfloats} % for two-column equations to appear at bottom
\usepackage{mathtools} % for \coloneqq
\usepackage{bbm,dsfont} % for double-struck 1
\usepackage{mathrsfs} % for mathscr font
\usepackage{url}

\def\BibTeX{{\rm B\kern-.05em{\sc i\kern-.025em b}\kern-.08em
    T\kern-.1667em\lower.7ex\hbox{E}\kern-.125emX}}
\makeatletter
\let\NAT@parse\undefined
\makeatother
\usepackage[hidelinks]{hyperref}

\DeclareMathOperator{\dist}{dist}
\DeclareMathOperator{\col}{col}

\newcommand{\vect}[1]{\mathbf{#1}}
\newcommand{\matr}[1]{\mathbf{#1}}
\newcommand{\tran}{\mathsf{T}}
\newcommand\eqm[1]{\bar{#1}}

\renewcommand{\Re}{\operatorname{Re}}
\renewcommand{\Im}{\operatorname{Im}}

\newtheorem{corollary}{Corollary}
\newtheorem{lemma}{Lemma}
\newtheorem{definition}{Definition}
\newtheorem{remark}{Remark}
\newtheorem{proposition}{Proposition}
\newtheorem{assumption}{Assumption}

\begin{document}
\title{Shifted Dissipativity and Rotor-Angle Feedback for Orbital Stability of Power Networks}
\author{Xinyuan Jiang, \IEEEmembership{Member, IEEE}
%\thanks{This paragraph of the first footnote will contain the date on which you submitted your paper for review. }
\thanks{X. Jiang is an independent researcher (e-mail: \url{j_jxy@outlook.com}).
}
%\thanks{X. Jiang is an independent researcher (e-mail: j\_jxy@outlook.com).}
}

\maketitle
\thispagestyle{empty}

\begin{abstract}
This paper considers a shifted dissipativity property of synchronous generators (SGs) with the aim of establishing stability of the synchronous orbit of a power network. Existing shifted-passivity analyses are formulated in a reference frame attached to the rotor and therefore suppress the problem of rotor-angle alignment in a multimachine network. The analysis takes two steps to reconcile angle-modulated power sources with dissipation. First, we define a weaker shifted dissipativity property whose supply rate retains a cross term between rotor-angle and terminal-voltage errors. We show that an SG connected to an infinite bus is locally asymptotically stable if it satisfies this property. Second, for networks of multiple SGs interconnected through dynamic transmission lines and static loads, we introduce a consensus-like rotor-angle feedback that consolidates the contributions of the distributed angle-modulated power sources into a single effective channel. This channel is then dominated by the dissipative dynamics of an auxiliary state variable representing the phase of the synchronous orbit, thereby establishing local orbital stability. Numerical examples demonstrate the feasibility of the shifted dissipativity condition.
\end{abstract}


\section{Introduction}
\label{sec:introduction}

\IEEEPARstart{T}{his} paper aims to define and substantiate a shifted dissipativity property for SGs that guarantees the orbital stability of a power network in a compositional manner.
The standard definition of dissipativity characterizes an input--output system in terms of energy dissipation with respect to a generalized supply rate. It can be used to determine whether a system will be stable when connected to a steady external power source. This applies, e.g., to DC power systems. However, in AC power systems, the orientations of the power sources are determined by the SG rotor angles, which are not synchronized with one another a priori. In an SG, the power source is the input mechanical torque, which is steady in the $dq$ reference frame that rotates with the rotor but not in a stationary reference frame shared by every component of the power system. The instantaneous power that the SG exports depends greatly on the rotor angle. Therefore, the stability of an AC system crucially depends on whether the rotor frequencies converge to a common value. The standard definition of dissipativity provides no obvious mechanism for accounting for this angle behavior. To establish dissipativity when the power transfer through a port is nonzero, the power sources must typically be hidden or separated from the port variables. Misaligned rotor angles, however, reveal that an SG is inherently a power source.

The literature reflects a dilemma between dissipativity and angle alignment. In this review, we also include literature on inverter-based resources (IBRs) because the same principles apply. In a series of works on the multivariable cell structure~\cite{barabanov2017conditions,schiffer2019global}, Schiffer and colleagues observed that, under the assumption that the storage function of the overall power network is positive definite, the channels of rotor-angle differences can have no associated dissipation. This prompted the development of the Leonov function approach to prove stability of systems with a periodic steady state. The same observation applies to each individual SG: the rotor angle can have no dissipation associated with it. The problem is that, since the rotor angle is tied to the orientation of the power source, dissipation in other channels, whether from the control loops or the parasitic losses, cannot dominate the power source. Ref.~\cite{spanias2018system} provides a counterexample to this approach: the authors consider shifted passivity with respect to a reference trajectory in the stationary reference frame. However, Remark~\ref{lem_supply} proves that an SG cannot satisfy this property, even before the absence of angle dissipation is considered.

Another group of studies abandons frequency synchronization and the issue of angle alignment~\cite{riverso2014plug,6975228,tucci2020scalable,strehle2021unified}. Instead, they control the angles of the SGs to match a given set of references, thus turning the problem into a classical nonlinear output regulation problem. In IBR control, this approach is called angle droop. However, this approach is difficult to implement beyond simulations. The practical implementation requires every generator to maintain a synchronized clock~\cite{kolluri2018stability}. While this may be readily achieved using GPS, the approach also requires precise knowledge of the steady state to avoid issues such as integrator wind-up. The latter can pose serious challenges if the steady state is always changing due to external disturbances. Moreover, with this approach, adjusting each generator's power output to meet the power demands relies solely on prescribed angle references, which in turn depend on a communication network with minimal delay. However, in IBR control, there is a particular emphasis on power sharing through the frequency-synchronization mechanism, which requires the IBRs to adjust their output voltage frequencies based on the instantaneous power measurements rather than on commands issued by a centralized controller over a communication network.

In view of these difficulties, contemporary works tend to establish dissipativity only for two simplified setups~\cite{natarajan2014almost,monshizadeh2019conditions,Chen2024Unified}. The first setup is a single SG connected to an infinite bus (SMIB). Compared with that of a multimachine power network, the steady state of the SMIB system is a relative equilibrium, i.e., a steady state that can be transformed into an equilibrium with a possibly time-dependent coordinate transformation. The second setup is to prove dissipativity in the $dq$ reference frame attached to the rotor and to regard the interaction with everything else in the power network as a disturbance. This setup not only reduces the steady state to an equilibrium but also reduces the power source to a constant bias in the equations~\cite{monshizadeh2019conditions}. Indeed, much of the theory surrounding SG dynamics focuses on the dynamics in the $dq$ reference frame, with examples such as the one- and two-axis models~\cite{anderson2004power}.
For example, 
the analysis in~\cite{caliskan2014compositional} incorrectly assumed that dissipativity in the $dq$ reference frame extended to a stationary reference frame. Since the issue of angle alignment was neglected, direct angle control was subsequently required in the correction~\cite{caliskan2016correction}.

To the best of our knowledge, only a limited number of works study the stationary-frame SG dynamics, which constitutes a significant gap in power system stability analysis. The traditional way to bridge this gap is to reduce the transmission-line dynamics to a static power flow model. With the assumption of zero line resistance, the model can be reduced so that its state variables consist only of rotor frequencies and angles. This assumption suppresses the distinction between absolute-angle dynamics and relative-angle dynamics. 
The standard energy function method and its variants rely on this assumption to find a local storage function~\cite{schiffer2014conditions,schiffer2019global,de2018bregman,nishino2025necessary}. 
However, constructing an analogous local storage function for a model that includes resistive, dynamic transmission lines appears disproportionately difficult~\cite{gross2019effect,gao2025explicit}.

In this paper, we build on the shifted passivity concept in~\cite{monshizadeh2019conditions} and propose a shifted dissipativity condition for SGs that addresses the issue of angle alignment. The proposed supply rate retains a cross term between the rotor angle and the terminal voltage. This makes the dissipativity condition feasible because the cross term accounts for the power-source contribution caused by rotor-angle misalignment. 

As the first result, we show that satisfying this condition implies that an SG is locally asymptotically stable in the SMIB setup, and we numerically validate the feasibility of this condition with realistic SG parameters.

As the second result, we prove that when the shifted dissipativity condition is satisfied for each SG in a power network with dissipative dynamic transmission lines, we can apply a consensus-like rotor-angle feedback to stabilize the periodic orbit of the power network. Compared with the single-machine result, the network result involves a periodic orbit that cannot be reduced to an equilibrium point without losing dynamical features that are important for the frequency synchronization mechanism as well as for power sharing. 
Existing energy-function methods suppress these dynamical features because their simplified models retain static diffusive coupling instead of dynamic, dissipative transmission-line dynamics.

Our approach to resolving the dilemma between angle alignment and dissipation is to introduce an auxiliary state variable to represent the phase of the reference trajectory. We endow this auxiliary state variable with negative-feedback dynamics, thereby creating an additional channel of dissipation. The proposed rotor-angle feedback, introduced as a consensus-like current injection at each SG terminal, consolidates the distributed power sources associated with the individual SGs into a single channel that can be dominated by the dissipation associated with the auxiliary state variable, thereby enabling us to prove convergence of the extended state vector to the extended periodic orbit.

The rest of the paper is organized as follows. Section~\ref{sec_set_up} introduces notation and defines the synchronous steady state. Section~\ref{sec_SG} introduces the synchronous generator and investigates its shifted dissipativity property. Section~\ref{sec_grid} completes the remaining parts of the power system model and proves the network stability implications of the shifted dissipativity conditions. Numerical examples are presented in Section~\ref{sec_numerical} to verify the shifted dissipativity conditions. 
Section~\ref{sec_concl} concludes the paper.
\section{Problem Setup} \label{sec_set_up}

\subsection{Notation}
We use $\mathbb{T} = \mathbb{R}/(2\pi \mathbb{Z})$ for the one-dimensional torus, $\mathbb{S}^1 = \{e^{j\theta} \mid \theta \in \mathbb{T}\}$ for the unit circle in $\mathbb{C}$, where $j $ is the imaginary unit.
We use $\mathbb{R}_{\geq 0}$ for the set of non-negative real numbers, $\Re\{\cdot\}$ and $\Im\{\cdot\}$ for the real and imaginary parts, and $\mathbf{L}^2_{\mathrm{loc}}$ for locally square integrable functions.

For a set of vectors $\vect x_i,\, i = 1,\ldots, n$, denote by $\col(\vect x_i)_{i=1}^n$ the vector that stacks $\vect x_i$. The vector of $n$ ones is $\mathbbm{1}_n$; the zero vector of size $n$ is $\mathbf{0}_n$; the identity matrix of size $n$ is $\matr I_n$; the zero matrix of size $n\times m$ is $\matr 0_{n\times m}$. 
Denote by $(\cdot)^*$ the Hermitian transpose operator, and $(\cdot)^\tran$ the transpose.

For complex vectors, \(\langle \vect u, \vect v\rangle\coloneqq\Re\{\vect u^* \vect v\}\)
denotes the real inner product. The gradient of a real-valued continuously differentiable function
on a complex vector space is understood with respect to this real inner
product; equivalently, for $f: \mathbb{C}^n \to \mathbb{R}$, 
$\nabla f(\vect x) = 2 \begin{bmatrix}
    \frac{\partial f(\vect x)}{\partial x_1^*} &\cdots &\frac{\partial f(\vect x)}{\partial x_n^*}
\end{bmatrix}^\tran$,
where each $\frac{\partial f(\vect x)}{\partial x_i^*}$ is the Wirtinger derivative~\cite[Section~1.4]{remmert1991theory}; the factor $2$ originates from the convention for Wirtinger derivatives. It holds that
$
\frac{\partial f}{\partial \vect x}(\vect x) \vect v = \langle \nabla f(\vect x), \vect v \rangle
$.
We use $\matr M \succ 0$ to denote that all eigenvalues of an implicitly Hermitian matrix $\matr M = \matr M^* \in \mathbb{C}^{n\times n}$ are positive.

Let $\mathcal M \subset \mathbb{C}^n$ be a Riemannian manifold~\cite[Chapter~3]{bullo2005geometric} embedded in $\mathbb{C}^n$.
For any set $U \subset \mathcal M$ and $r > 0$, we denote by
$B_{\mathcal M}(U, r) \coloneqq \big\{ \vect x \in \mathcal M \mid \dist_{\mathcal M}(\vect x, U) < r \big\}$
the open $r$-neighborhood of $U$ relative to $\mathcal M$.
The big O notation $f(\vect x) = O\!\left(\|\vect x\|^k\right)$ with $k\in \mathbb{N}$ means that there exist $c > 0$ and $\delta > 0$ such that $|f(\vect x)| < c\|\vect x\|^k$ for all $\vect x \in B_{\mathcal M}(\vect 0_n, \delta)$.

\subsection{Stability of a Synchronous Orbit}

Assuming constant input disturbances (e.g., load and generation setpoints), we abstract the multimachine power system as an autonomous ODE system 
\begin{equation} \label{E:system}
    \dot{\vect x} = F(\vect x)
\end{equation}
where $F : \mathbb{X} \to T\mathbb{X}$ is a smooth vector field on a Riemannian manifold $\mathbb{X} \subset \mathbb{C}^n$~\cite[Chapter~3]{bullo2005geometric}. The abstract model \eqref{E:system} is suitable only if the algebraic constraints, e.g., Kirchhoff Voltage Law (KVL) and Kirchhoff Current Law (KCL), are solvable. Later in Section~\ref{sec_grid}, we will explicitly build a relatively general power system model for which KCL and KVL are indeed solvable.
We denote by $\Phi_F(t, \vect x_0)$ the flow map of \eqref{E:system}, which advances each initial condition $\vect x_0 \in \mathbb{X}$ by a given time $t\geq 0$ along the trajectory.

We assume that \eqref{E:system} admits a single-frequency sinusoidal solution $\eqm{\vect x}(t)$ with period $T > 0$ so that $\eqm{\vect x}(t) = \eqm{\vect x}(t + T)$ for all $ t \geq 0$. The angular frequency is $\eqm\omega = 2\pi/T > 0$.\footnote{A single-frequency sinusoidal solution of a power system is also known as a synchronous steady state~\cite{gross2018steady}. See recent results on the existence problem in~\cite{lorenzetti2021equilibrium,dvijotham2015differential}.} 
Without this assumption, essentially all subsequent analyses break. In this paper, we address the stability of the dynamics, and regard the existence of synchronous steady states as an independent algebraic problem for future work.

\begin{definition}
A \emph{synchronous steady state} is a residual motion, as determined by the dynamics of the power system, where all voltages and currents, in either $abc$ or $\alpha\beta0$ coordinates, rotate at the same frequency $\eqm\omega$~\cite[Subsection~2.3]{gross2018steady}.
\end{definition}

\begin{assumption} \label{assum_exist}
Assume the power system \eqref{E:system} that is being studied has a synchronous steady state.
\end{assumption}

Define the orbit associated with the periodic solution $\eqm{\vect x}(t)$ as 
\begin{equation*}
\Gamma \coloneqq \{\eqm{\vect x}(t)\mid t \in [0, T] \}.
\end{equation*}
The set $\Gamma \subset \mathbb{X}$ is a submanifold that is diffeomorphic to $\mathbb{T}$. We parameterize each point in $\Gamma$ as
\begin{equation} \label{E:parameter}
    \eqm{\vect x}[\varphi] \coloneqq \eqm{\vect x}(\varphi/\eqm\omega),\quad \forall \varphi \in \mathbb{T}.
\end{equation} 
Based on the periodic-orbit parameterization \eqref{E:parameter}, the original periodic solution $\eqm{\vect x}(t)$ is expressed as $\eqm{\vect x}(t) = \eqm{\vect x}[\varphi(t)]$
with
\begin{equation*}
    \varphi(t) = \eqm\omega t.
\end{equation*}
Any periodic solution on $\Gamma$ with a different initial condition can be expressed as $\eqm{\vect x}(t) = \eqm{\vect x}[\varphi(t)]$
with
\begin{equation*}
    \varphi(t) = \eqm\omega t + \varphi_0,\quad \varphi_0\in \mathbb{T}.
\end{equation*}
Therefore, using the periodic-orbit parameterization \eqref{E:parameter}, we obtain a bijective mapping from the residual motion on $\Gamma$ to the following one-dimensional dynamics on $\mathbb{T}$~\cite{wieland2010,isidori2014robust}:
\begin{equation} \label{E:residual}
    \dot\varphi(t) = \eqm\omega,\quad \varphi(0) = \varphi_0 \in \mathbb{T}.
\end{equation}
Our objective in this paper is to find a sufficient condition for the orbit $\Gamma$ to be attractive because the asymptotic phase does not matter; see \cite{manchester2014transverse} for more background.
The precise definition of orbital attraction is given below.

\begin{definition}[Definition~2.12.2 of~\cite{bhatia2006dynamical}] \label{defn_stability1}
A closed set $\Gamma \subset \mathbb{X}$ is said to be \emph{uniformly stable} if, for each $\varepsilon > 0$, there exists $\delta = \delta(\varepsilon) > 0$ such that
\begin{equation*}
    \Phi_F\!\left(\mathbb{R}_{\geq 0}, B_{\mathbb{X}}(\Gamma, \delta)\right) \subset B_{\mathbb{X}}(\Gamma, \varepsilon).
\end{equation*}
\end{definition}
\vspace{\topsep}

\begin{definition}[Definition~2.12.12 of~\cite{bhatia2006dynamical}] \label{defn_stability2}
A closed set $\Gamma \subset \mathbb{X}$ is an \emph{attractor} if there exists $\delta > 0$ such that for each $\vect y \in B_{\mathbb{X}}(\Gamma, \delta)$, $\dist_{\mathbb{X}}(\Phi_F(t, \vect y), \Gamma) \to 0$ as $t\to\infty$. \\
It is said to be \emph{asymptotically stable} if it is uniformly stable and is an attractor.
\end{definition}

\section{Shifted Dissipativity of Synchronous Generators} \label{sec_SG}

\subsection{Dynamical Model}

As part of the synchronous generator (SG) model that we consider in this paper~\cite{levron2018tutorial}, the second-order swing equation for the rotor dynamics is written in state-space form as\footnote{The subscript $\texttt e$ is the index for the SG. In the first half of this paper, we omit the subscript $\texttt e$ for the individual state variables to ease notation.}
\begin{subequations} \label{E:swing} \begin{align}
    J \dot \omega &= -D \omega + T_m + \left\langle I, j\lambda e^{j\theta}\right\rangle, \label{E:swing1} \\
    \dot\theta &= \omega,
\end{align} \end{subequations}
where $J> 0$ is the moment of inertia, $D > 0$ is the viscous damping constant, $T_m > 0$ is the mechanical torque, $\lambda > 0$ is the magnetic flux linkage, $\omega \in \mathbb{R}$ is the electrical rotor frequency,\footnote{That is, the mechanical rotor frequency multiplied by the number of pole pairs.} $\theta \in \mathbb{T}$ is the electrical rotor angle, 
and $I \in \mathbb{C}$ is the space vector for the stator current.\footnote{Current that flows \emph{into} the machine is \emph{positive}. This is consistent with the pH picture in~\cite{fiaz2013port,van2014port} where each port absorbs power under the adopted sign convention when the inner product of its input and output is positive.}\footnote{The real and imaginary parts of $I$ are the $\alpha$- and $\beta$-components of the $\alpha\beta0$ transform of the $abc$ stator current, respectively~\cite{o2019geometric}.} 

\begin{assumption}
Assume that the rotor construction is rotationally symmetric and that the three phases of the stator windings are invariant under permutations.
\end{assumption}

The stator-circuit equation is written as
\begin{equation} \label{E:stator}
    L \dot I = -R I - \lambda j \omega e^{j\theta} + V,
\end{equation}
where $L > 0$ is the stator self-inductance, $R> 0$ is the stator resistance, and $V \in \mathbb{C}$ is the space vector for the terminal voltage.

\begin{assumption}
Assume that the magnetic flux linkage $\lambda > 0$ and the mechanical torque $T_m > 0$ are positive constants. 
\end{assumption}

\begin{remark} 
If all rotor states are observable, then a constant $\lambda$ can be implemented with a special excitation voltage function; see the controller just above Assumption~1 in~\cite{caliskan2014compositional}.
Compared with the full-order SG model, we have essentially omitted the $\lambda$ dynamics, which include the dynamics of the excitation circuits and the damper windings. Thus, the SG model \eqref{E:swing} and \eqref{E:stator} represents an idealized case with infinite-bandwidth control of the flux linkage $\lambda$.

To simplify notation, we embed the rotor angle $\theta$ into the complex variable $\xi \coloneqq e^{j\theta}$, which maps from the $1$-torus $\mathbb{T}$ to the unit circle $\mathbb{S}^1$. We choose the state vector for the SG to be
\begin{equation*}
    \vect x_{\texttt e} \coloneqq \begin{bmatrix}
        j \omega \xi &\xi &I
    \end{bmatrix}^\tran, \quad \texttt e= 1,\ldots , \texttt g,
\end{equation*}
which is a subset of the full power-system state vector $\vect x$ for \eqref{E:system}.
The SG state manifold is
\begin{equation*}
    \mathbb{X}_{\texttt e} \coloneqq \left\{\begin{bmatrix}
        x_1 &x_2 &x_3
    \end{bmatrix}^\tran \in \mathbb{C}\times \mathbb{S}^1\times \mathbb{C} \mid x_1/(jx_2) \in \mathbb{R} \right\}.
\end{equation*}
We regard $\vect u_{\texttt e} = V$ as the input to the SG, and $\vect y_{\texttt e} = I$ as the output. These variables define the port through which the SG subsystem
interacts with the remainder of the power system \eqref{E:system}, and both are
components of the full power-system state \(\vect x\).
\end{remark}

\begin{remark} \label{prop_power}
If the real and reactive power outputs through the SG port are positive in the synchronous steady state, i.e., 
$\eqm P \coloneqq \Re\!\left\{\eqm I[0]^* \eqm V[0] \right\} <0$ and $\eqm Q \coloneqq \Im\!\left\{\eqm I[0]^* \eqm V[0]\right\} < 0$, then $\left\langle \eqm I[0], j \eqm\xi[0] \right\rangle < 0$ and $\left\langle \eqm I[0], \eqm\xi[0] \right\rangle < 0$.
\end{remark}
%\vspace{4pt}

Due to the space limitation, the proof is omitted.

Define the constant
\begin{equation*}
    \mathcal K \coloneqq \lambda \left\langle \eqm I[0], \eqm\xi[0] \right\rangle.
\end{equation*}
As we will see later, the stability results in this paper apply to operating conditions where $\mathcal K > 0$, hence the SG absorbs steady-state reactive power. The reactive power demands are assumed to be supplied by capacitor banks, either physical or simulated by power electronics.

\subsection{Definition of Shifted Dissipativity}
 
Passivity in the standard guise is infeasible for the SG 
because of the steady-state power export through its port~\cite{monshizadeh2019conditions,van2000l2}; see Remark~\ref{lem_supply}. Instead, we are concerned with dissipativity with respect to the supply rate
\begin{align} \label{E:supply}
    s(\vect u_{\texttt e}, \vect y_{\texttt e}, \varphi) &\coloneqq \frac{\lambda}{L} \left\langle \xi - \eqm\xi[\varphi], \vect u_{\texttt e} - \eqm{\vect u}_{\texttt e}[\varphi] \right\rangle \notag \\
    &\:\quad\, + \left\langle \vect y_{\texttt e} - \eqm{\vect y}_{\texttt e}[\varphi] , \vect u_{\texttt e} - \eqm{\vect u}_{\texttt e}[\varphi]\right\rangle.
\end{align} 

\begin{definition} \label{defn_1}
The SG system \eqref{E:swing} and \eqref{E:stator} is said to be locally shifted dissipative with respect to the supply rate $s(\vect u_{\texttt e}, \vect y_{\texttt e}, \varphi)$ in \eqref{E:supply} if there exists a neighborhood $U \subset \mathbb{X}_{\texttt e}\times \mathbb{T}$ of the extended orbit 
$\Gamma_{\mathrm{ext}} = \left\{(\eqm{\vect x}_{\texttt e}[\varphi], \varphi) \mid \varphi \in \mathbb{T} \right\}$
and a nonnegative function $b: U \to \mathbb{R}_{\geq 0}$ such that, along the flow of the extended system consisting of \eqref{E:swing}, \eqref{E:stator}, and some given dynamics for $\varphi$,
\begin{align} 
    &\int_0^{t} s(\vect u_{\texttt e}(\tau), \vect y_{\texttt e}(\tau), \varphi(\tau))\, d \tau \geq -b(\vect x_{\texttt e}(0), \varphi(0)), \notag \\
    &\forall (\vect u_{\texttt e}(\cdot) - \eqm{\vect u}_{\texttt e}[\varphi(\cdot)]) \in \mathbf{L}^{2}_{\mathrm{loc}}(\mathbb{R}_{\geq 0}; \mathbb{C}) 
    \;\; \text{and}\;\; (\vect x_{\texttt e}(0), \varphi(0)) \in U \notag \\
    &\quad\text{such that} \;\; (\vect x_{\texttt e}(\tau), \varphi(\tau)) \in U,\, \tau \in [0, t]. \label{E:shifted_passivity}
\end{align}
\end{definition}
\vspace{\topsep}

Definition~\ref{defn_1} specifies the required asymptotic behavior for the shifted input and output. 
For the subsequent analysis, we will use the following equivalent definition.

\begin{definition} \label{defn_2}
The SG system \eqref{E:swing} and \eqref{E:stator} is said to be 
locally shifted dissipative with respect to the supply rate $s(\vect u_{\texttt e}, \vect y_{\texttt e}, \varphi)$ in \eqref{E:supply} if there exists a neighborhood $U \subset \mathbb{X}_{\texttt e}\times \mathbb{T}$ of the extended orbit 
$\Gamma_{\mathrm{ext}} = \left\{(\eqm{\vect x}_{\texttt e}[\varphi], \varphi) \mid \varphi \in \mathbb{T} \right\}$
and a lower-bounded function $S:U \to \mathbb{R}$ such that, along the flow of the extended system consisting of \eqref{E:swing}, \eqref{E:stator}, and some given dynamics for $\varphi$,
    \begin{align} 
        &S(\vect x_{\texttt e}(t), \varphi(t)) \leq S(\vect x_{\texttt e}(0), \varphi(0)) \notag \\
        &\qquad\qquad\qquad\quad  + \int_0^t s(\vect u_{\texttt e}(\tau), \vect y_{\texttt e}(\tau), \varphi(\tau))\, d \tau, \notag \\
        &\forall (\vect u_{\texttt e}(\cdot) - \eqm{\vect u}_{\texttt e}[\varphi(\cdot)]) \in \mathbf{L}^{2}_\mathrm{loc}(\mathbb{R}_{\geq 0}; \mathbb{C}) 
        \;\; \text{and}\;\; (\vect x_{\texttt e}(0), \varphi(0)) \in U \notag \\
        &\quad \text{such that} \;\; (\vect x_{\texttt e}(\tau), \varphi(\tau)) \in U,\, \tau \in [0, t]. \label{E:shifted_passivity2}
    \end{align}
\end{definition}
\vspace{\topsep}

Note that Definition~\ref{defn_2} is quite general---it does not require the storage function $S(\vect x_{\texttt e}, \varphi)$ to be differentiable.

\begin{lemma} \label{lem_equiv}
Definitions \ref{defn_1} and \ref{defn_2} are equivalent.
\end{lemma}

For completeness, the proof is given in Appendix~\ref{proof_equiv}.

\begin{remark} \label{lem_supply}
Assume that the SG exports real power at the steady state, i.e, $\langle \eqm{\vect y}_{\texttt e}[0], \eqm{\vect u}_{\texttt e}[0] \rangle < 0$.
Then, with the standard phase dynamics $\dot\varphi = \eqm\omega$, it is infeasible to prove locally shifted passivity, i.e., being 
locally shifted dissipative with respect to the supply rate $s_{\mathrm{sp}}(\vect u_{\texttt e}, \vect y_{\texttt e}, \varphi) \coloneqq \langle \vect y_{\texttt e} - \eqm{\vect y}_{\texttt e}[\varphi], \vect u_{\texttt e} - \eqm{\vect u}_{\texttt e}[\varphi] \rangle$. %i.e., locally shifted passive.
\end{remark}

\begin{proof}
Let us consider the trajectory $\vect x_{\texttt e} = \eqm{\vect x}_{\texttt e}[\varphi + \Delta]$ for a nonzero $\Delta$ arbitrarily close to zero. Clearly, it belongs to a neighborhood of the extended orbit---its distance to the reference $\eqm{\vect x}_{\texttt e}[\varphi]$ remains equal to $\Delta$ for all $t \geq 0$. However, note that
\begin{align*}
    &\quad\,\, s_{\mathrm{sp}}(\eqm{\vect u}_{\texttt e}[\varphi + \Delta], \eqm{\vect y}_{\texttt e}[\varphi + \Delta], \varphi) \\
    &= \big\langle \eqm{\vect y}_{\texttt e}[\varphi + \Delta] - \eqm{\vect y}_{\texttt e}[\varphi], \eqm{\vect u}_{\texttt e}[\varphi + \Delta] - \eqm{\vect u}_{\texttt e}[\varphi] \big\rangle \\
    &= \left\langle (e^{j\Delta} - 1) \eqm{\vect y}_{\texttt e}[\varphi], (e^{j\Delta} - 1) \eqm{\vect u}_{\texttt e}[\varphi] \right\rangle \\
    &= 2(1-\cos(\Delta)) \left\langle \eqm{\vect y}_{\texttt e}[\varphi], \eqm{\vect u}_{\texttt e}[\varphi] \right\rangle ,
\end{align*}
which is negative for $\Delta \in (-\frac{\pi}{2}, \frac{\pi}{2} )$. Hence, by Definition~\ref{defn_1}, the SG system cannot be locally shifted passive, i.e., it cannot be locally shifted dissipative with respect to $s_{\mathrm{sp}}(\vect u_{\texttt e}, \vect y_{\texttt e}, \varphi)$.
\end{proof}

\begin{remark}
It is well-known that passivity requires that the system to absorb power as time goes to infinity; that is, a passive system cannot be a power source. If the steady state is an equilibrium point for the input and output, then we can shift the dynamics to the equilibrium point so that the shifted supply rate becomes zero as time goes to infinity. This gives rise to the concept of shifted passivity \cite{monshizadeh2019conditions}. However, as we see in Remark~\ref{lem_supply}, if the steady state is a periodic orbit for the input and output, then we cannot simply shift away the asymptotic supply rate. This is the primary reason we choose to focus instead on the alternative supply rate \eqref{E:supply}.
\end{remark}

\subsection{Shifted Dissipativity Condition}

For the extended system consisting of \eqref{E:swing}, \eqref{E:stator}, and $\dot\varphi = \eqm\omega$, let us consider the candidate storage function defined as
\begin{align}
    H_{\texttt e}(\vect x_{\texttt e}, \varphi)
    &\coloneqq \frac{1}{2} J \|j \omega\xi - j\eqm\omega \eqm\xi[\varphi]\|^2 + \frac{1}{2} \mathcal K \|\xi - \eqm\xi[\varphi]\|^2 \notag \\
    &\:\, \quad + \frac{1}{2} L \left\| I + \frac{\lambda}{L} \xi - \left(\eqm I[\varphi] + \frac{\lambda}{L} \eqm\xi[\varphi] \right) \right\|^2 \label{E:H_sigma}
\end{align}
Note that the term $L I + \lambda \xi$ is the stator flux linkage.

\begin{proposition}[Energy balance for original storage] \label{prop_ineq}
Along the generator dynamics and $\dot\varphi = \eqm\omega$, it holds that
\begin{align} 
    \dot H_{\texttt e}(\vect x_{\texttt e}, \varphi) &= {-}\left\langle\vect x_{\texttt e} - \eqm{\vect x}_{\texttt e}[\varphi], \matr P_{\texttt e}(\omega, \xi, \varphi) (\vect x_{\texttt e} - \eqm{\vect x}_{\texttt e}[\varphi]) \right\rangle  \notag \\
    &\quad\,  + \frac{\lambda}{L} \left\langle \xi - \eqm\xi[\varphi] , \vect u_{\texttt e} - \eqm{\vect u}_{\texttt e}[\varphi]\right\rangle \notag \\
    &\quad\, + \left\langle \vect y_{\texttt e} - \eqm{\vect y}_{\texttt e}[\varphi], \vect u_{\texttt e} - \eqm{\vect u}_{\texttt e}[\varphi] \right\rangle, \label{E:eb_raw_raw}
\end{align}
where $\matr P_{\texttt e}(\omega, \xi, \varphi)$ is defined in \eqref{E:P}.
\end{proposition}

\begin{figure*}[t]
\begin{align}
    &\quad\,\,\matr P_{\texttt e}(\omega, \xi, \varphi) = \begin{bmatrix}
        D &-\frac{1}{2} \left(J\eqm\omega \omega + j D\eqm\omega \right)+ \frac{1}{4}\lambda \eqm I[\varphi]^*  (\xi - \eqm\xi[\varphi]) &-\frac{1}{2} \lambda \\[3pt]
        * &0 &\frac{1}{2} \left(\frac{\lambda R}{L} - j \lambda \eqm\omega \right)-\frac{1}{4} \lambda j\eqm\omega  \eqm\xi[\varphi]  (\xi - \eqm\xi[\varphi])^* \\[3pt]
        * &* &R
    \end{bmatrix} \label{E:P}
\end{align}
\hrulefill
\vspace*{4pt}
\end{figure*}

The proof is given in Appendix~\ref{proof_ineq}. 

Using
\begin{align}
    &\quad\,\, \|j\omega\xi - j \eqm\omega\eqm\xi[\varphi]\|^2 \notag \\
    &= (\omega - \eqm\omega)^2 + \eqm\omega\omega \|\xi - \eqm\xi[\varphi]\|^2 \notag \\ %\label{E:auxiliary} \\
    &= (\omega - \eqm\omega)^2 + \left\langle j\omega\xi - j\eqm\omega\eqm\xi[\varphi], j\eqm\omega (\xi - \eqm\xi[\varphi]) \right\rangle \notag \\
    &\quad\, +O\!\left(\|\vect x_{\texttt e} - \eqm{\vect x}_{\texttt e}[\varphi]\|^3\right), \notag 
\end{align}
and
\begin{equation*} %\label{E:expansion2}
    \xi-\eqm\xi[ \varphi]
    =
    j\eqm\xi[ \varphi] (\theta - \eqm\theta[\varphi]) +O\!\left(\|\vect x_{\texttt e} - \eqm{\vect x}_{\texttt e}[\varphi]\|^2\right),
\end{equation*}
we obtain the second-order part of \eqref{E:eb_raw_raw} as
\begin{align}
    &\quad\,\, \dot{H}_{\texttt e}(\vect x_{\texttt e}, \varphi) \notag \\
    &= -D(\omega - \eqm\omega)^2 + J\eqm\omega^2(\omega - \eqm\omega) (\theta - \eqm\theta[\varphi]) \notag \\
    &\quad\, + \lambda (\omega - \eqm\omega) \left\langle j \eqm\xi[\varphi], I - \eqm I[\varphi] \right\rangle \notag \\
    &\quad\, - \frac{\lambda R}{L} (\theta - \eqm\theta[\varphi]) \left\langle j\eqm\xi[ \varphi], I - \eqm I[\varphi] \right\rangle - R\|I - \eqm I[\varphi] \|^2 \notag \\
    &\quad\, + \frac{\lambda}{L} \left\langle \xi - \eqm\xi[\varphi] , \vect u_{\texttt e} - \eqm{\vect u}_{\texttt e}[\varphi]\right\rangle + \left\langle \vect y_{\texttt e} - \eqm{\vect y}_{\texttt e}[\varphi], \vect u_{\texttt e} - \eqm{\vect u}_{\texttt e}[\varphi] \right\rangle \notag \\
    &\quad\, + O\!\left(\|\vect x_{\texttt e} - \eqm{\vect x}_{\texttt e}[\varphi]\|^3\right). \label{E:eb_raw}
\end{align}
Eq. \eqref{E:eb_raw} falls short of proving local shifted dissipativity due to the lack of angle dissipation, i.e., $-(\theta - \eqm\theta)^2$. This is addressed by augmenting the storage function $H_{\texttt e}(\vect x_{\texttt e}, \varphi)$ with the cross term $(\omega - \eqm\omega)\left\langle \xi, j \eqm\xi[\varphi] \right\rangle$.

\begin{proposition}[Local definiteness of augmented storage] \label{lem_pos}
For $\sigma_{\texttt e} \in \mathbb{R}$, define the augmented storage
\begin{align}
    \widetilde H_{\texttt e}(\vect x_{\texttt e}, \varphi; \sigma_{\texttt e})
    &\coloneqq
    H_{\texttt e}(\vect x_{\texttt e}, \varphi)+\sigma_{\texttt e} J (\omega - \eqm\omega) \left\langle \xi, j \eqm\xi[\varphi]\right\rangle \notag \\
    &\quad\:\, + \frac{1}{2} (\sigma_{\texttt e} D - J\eqm\omega\omega) \| \xi- \eqm\xi[\varphi] \|^2. \label{E:augmented_storage}
\end{align}
If 
\begin{equation} \label{E:pos_condition}
\mathcal K + \sigma_{\texttt e} D - \sigma_{\texttt e}^2 J > 0, 
\end{equation}
then the function \(\widetilde H_{\texttt e}(\vect x_{\texttt e}, \varphi; \sigma_{\texttt e})\) is locally positive definite in $\vect x_{\texttt e}$ with respect to
\(\vect x_{\texttt e}=\eqm{\vect x}_{\texttt e}[\varphi]\).
\end{proposition}

\begin{proof}
Let
\begin{align}
    &\nu\coloneqq\omega-\eqm\omega, &&
    \delta\coloneqq\arg(\eqm\xi[\varphi]^*\xi), && \zeta \coloneqq I - \eqm I[\varphi], \label{E:notation} \\
    &\iota \coloneqq \xi - \eqm\xi[\varphi], \notag
\end{align}
where the branch of \(\arg(\cdot)\) is chosen locally near zero.
We can write
\begin{align*}
    \widetilde H_{\texttt e}(\vect x_{\texttt e}, \varphi; \sigma_{\texttt e}) 
    &= \frac{1}{2}J (\nu +\sigma_{\texttt e} \delta)^2 + \frac{1}{2} (\mathcal K + \sigma_{\texttt e}D - \sigma_{\texttt e}^2 J) \delta^2 \notag \\
    &\quad\, + \frac{1}{2} L\left\|\zeta + \frac{\lambda}{L} \iota \right\|^2 + O\!\left(\|\vect x_{\texttt e} - \eqm{\vect x}_{\texttt e}[\varphi]\|^3\right)
\end{align*}
Hence $\widetilde H_{\texttt e}(\vect x_{\texttt e}, \varphi; \sigma_{\texttt e})$ is locally positive definite if and only if \eqref{E:pos_condition} holds.
\end{proof}

\begin{proposition}[Energy balance for augmented storage] \label{lem_diss}
The time derivative of the augmented storage function \eqref{E:augmented_storage} along the generator dynamics with
\(\dot \varphi=\eqm\omega\) is written as
\begin{align}
    \dot{\widetilde H}_{\texttt e}(\vect x_{\texttt e}, \varphi; \sigma_{\texttt e}) 
    &=
    -\widetilde{\mathcal Q}_{\texttt e}
    +
    \frac{\lambda}{L}
    \left\langle
        \xi-\eqm\xi[ \varphi],
        \vect u_{\texttt e} - \eqm{\vect u}_{\texttt e}[\varphi]
    \right\rangle \notag \\
    &\quad\,+ \left\langle \vect y_{\texttt e}-\eqm{\vect y}_{\texttt e}[\varphi], \vect u_{\texttt e} - \eqm{\vect u}_{\texttt e}[\varphi] \right\rangle
    \notag \\
    &\quad\, + O\!\left(\|\vect x_{\texttt e} - \eqm{\vect x}_{\texttt e}[\varphi]\|^3\right), \label{E:eb_raw2}
\end{align}
where the second-order part $\widetilde{\mathcal Q}_{\texttt e}$ is locally positive definite in $\vect x_{\texttt e} - \eqm{\vect x}_{\texttt e}[\varphi]$ if 
\begin{align}
    &\matr M_{\texttt e} (\sigma_{\texttt e}) \coloneqq \begin{bmatrix}
        D - \sigma_{\texttt e} J &0 &\frac{\lambda}{2} \\[2pt]
        * &\sigma_{\texttt e} \mathcal K &-\frac{\lambda}{2} \left(\frac{R}{L} - \sigma_{\texttt e} \right) \\[2pt]
        * &* &R
    \end{bmatrix} \succ 0. \label{E:diss_condition}
\end{align}
\end{proposition}
\vspace{\topsep}

\begin{proof}
From \eqref{E:eb_raw},
the second-order part of the generator dissipation associated with the original storage
\(H_{\texttt e}(\vect x_{\texttt e}, \varphi)\) is expressed in the notation of \eqref{E:notation} as
\begin{align*}
    \mathcal Q_{\texttt e}(\nu, \delta, \zeta)
    &=
    D\nu^2
    -
    J\eqm\omega^2\nu\delta
    -
    \lambda\nu
    \left\langle
        j\eqm\xi[ \varphi],\zeta
    \right\rangle                                      \\
    &\quad\, +
    \frac{\lambda R}{L}
    \delta \left\langle
        \zeta,j\eqm\xi[ \varphi]
    \right\rangle
    +
    R\|\zeta\|^2 .
\end{align*}
The first-order part of the swing equation \eqref{E:swing} is
\begin{align*}
J\dot\nu
&=
-D\nu
-
\mathcal K \delta
+
\lambda
\big\langle \zeta,j\eqm\xi[\varphi]\big\rangle +
O(\|\vect x_{\texttt e} - \eqm{\vect x}_{\texttt e}[\varphi]\|^2), \\
\dot\delta
&=
\nu.
\end{align*}
Hence the time derivative of 
$\mathcal Z\coloneqq \sigma_{\texttt e}J\nu \sin(\delta) + (\sigma_{\texttt e} D - J\eqm\omega^2) (1 - \cos(\delta))$
is
\begin{align*}
    \dot{\mathcal Z}
    &=
    J\eqm\omega^2 \nu\delta
    -
     \sigma_{\texttt e}\mathcal K\mathcal\delta^2
    +
     \sigma_{\texttt e}\lambda
    \delta \left\langle
        \zeta,j\eqm\xi[ \varphi]
    \right\rangle
    +
     \sigma_{\texttt e}J\nu^2 \\
    &\quad\, + O\!\left(\|\vect x_{\texttt e} - \eqm{\vect x}_{\texttt e}[\varphi]\|^3\right).
\end{align*}
Since
$\widetilde H_{\texttt e}(\vect x_{\texttt e}, \varphi; \sigma_{\texttt e})=H_{\texttt e}(\vect x_{\texttt e}, \varphi)+\mathcal Z$, we have
\[
    \widetilde{\mathcal Q}_{\texttt e}(\nu, \delta, \zeta)
    =
    \mathcal Q_{\texttt e}(\nu, \delta, \zeta)-(\dot{\mathcal Z})_\mathrm{quad},
\]
where $(\dot{\mathcal Z})_\mathrm{quad}$ is the second-order part of $\dot{\mathcal Z}$.
This gives
\begin{align}
    \widetilde{\mathcal Q}_{\texttt e}(\nu, \delta, \zeta) 
    &=
    (D-\sigma_{\texttt e} J)\nu^2
    %+
    %(\sigma_{\texttt e} D- \sigma_{\texttt e}^2 J-J\eqm\omega^2)\nu\delta
    +
    \sigma_{\texttt e} \mathcal K \delta^2 +
    R\|\zeta\|^2
     \notag \\
    &\quad\, +
    \lambda
    \left[
        \left(\frac{R}{L}-\sigma_{\texttt e}\right)\delta
        -\nu
    \right]
    \left\langle
        \zeta,j\eqm\xi[ \varphi]
    \right\rangle . \label{E:dot_Q}
\end{align}
The right-hand side can be expressed as
\begin{equation*}
    \begin{bmatrix}
        \nu \\
        \delta \\
        \zeta
    \end{bmatrix}^* \widetilde{\matr P}_{\texttt e}(\varphi, \sigma_{\texttt e}) \begin{bmatrix}
        \nu \\
        \delta \\
        \zeta
    \end{bmatrix},
\end{equation*}
where
\begin{align*}
    &\widetilde{\matr P}_{\texttt e}(\varphi, \sigma_{\texttt e}) = \begin{bmatrix}
        D - \sigma_{\texttt e} J &0 &\frac{\lambda}{2} j \eqm\xi[\varphi]^* \\[3pt]
        * &\sigma_{\texttt e} \mathcal K &-\frac{\lambda}{2} \left(\frac{R}{L} - \sigma_{\texttt e} \right) j \eqm\xi[\varphi]^* \\[3pt]
        * &* &R
    \end{bmatrix}.
\end{align*}
Since $R> 0$ and $\|\eqm\xi[\varphi]\| = 1$, we take the Schur complements of the last diagonal elements to verify
\begin{equation*}
    \matr M_{\texttt e}(\sigma_{\texttt e}) \succ 0 \quad \Longleftrightarrow \quad \widetilde{\matr P}_{\texttt e}(\varphi, \sigma_{\texttt e}) \succ 0.
\end{equation*}
This proves that $\widetilde{\mathcal Q}_{\texttt e}$ is positive definite in $(\nu, \delta, \zeta)$ and dominates the third-order remainder in a small neighborhood of $\eqm{\vect x}_{\texttt e} = \eqm{\vect x}_{\texttt e}[\varphi]$.
\end{proof}

\begin{remark} \label{rem_pos}
If \eqref{E:diss_condition} holds for some $\sigma_{\texttt e} \in \mathbb{R}$, then  \eqref{E:pos_condition} is satisfied if and only if $\mathcal K > 0$, and consequently $\sigma_{\texttt e} > 0$.
\end{remark}

\begin{proof}
\emph{Step 1 ($\Longleftarrow$).} Since the diagonal elements of $\matr M_{\texttt e}(\sigma_{\texttt e})$ must be positive for \eqref{E:diss_condition} to hold, we have $\sigma_{\texttt e} \mathcal K > 0$ and $D - \sigma_{\texttt e} J > 0$. If $\mathcal K > 0$, then $\sigma_{\texttt e} > 0$. This implies $\sigma_{\texttt e}(D - \sigma_{\texttt e} J ) > 0$, which together with $\mathcal K > 0$ implies \eqref{E:pos_condition}.

\emph{Step 2 ($\Longrightarrow$).}  Suppose \eqref{E:pos_condition} and \eqref{E:diss_condition} both hold with $\mathcal K < 0$. Then $\sigma_{\texttt e} < 0$ and $\sigma_{\texttt e} (D - \sigma_{\texttt e} J) < 0$. This contradicts with \eqref{E:pos_condition}. Hence \eqref{E:pos_condition} together with \eqref{E:diss_condition} implies $\mathcal K > 0$.
\end{proof}

\begin{remark}
From Proposition~\ref{lem_diss}, the  added cross term in the storage produces the angle dissipation $\sigma_{\texttt e} \mathcal K \delta^2$. Ref.~\cite{barabanov2017conditions} and the sequel \cite{schiffer2019global} take a very different approach to create angle dissipation. The authors construct a special function, called Leonov function, which is negative definite in $\delta$ and decays exponentially along the trajectory. This function implies the boundedness of the unwound angle $\delta$ in $\mathbb{R}$. Combined with an energy function in the traditional sense, it implies almost global stability. In comparison, our local stability analysis allows multiple locally asymptotically stable steady states, as illustrated in the SMIB example in Fig.~\ref{fig:certificate_sweeps}. Moreover, the power-system model in our analysis includes the transmission-line dynamics, while existing works usually assume lossless and quasi-steady-state lines~\cite{gross2019effect}.
\end{remark}

\begin{corollary} \label{cor}
An SMIB system given by \eqref{E:swing} and \eqref{E:stator} with $V = \eqm V[\varphi]$ is locally asymptotically stable if there exists $\sigma_{\texttt e} > 0$ that verifies condition \eqref{E:diss_condition}.
\end{corollary}

\begin{proof}
Conditions \eqref{E:pos_condition} and \eqref{E:diss_condition} with
\begin{equation*}
    \vect u_{\texttt e} - \eqm{\vect u}_{\texttt e}[\varphi] = V - \eqm V[\varphi] = 0
\end{equation*}
imply that $\widetilde H_{\texttt e}(\vect x_{\texttt e}, \varphi; \sigma_{\texttt e})$ is a local Lyapunov function for the error dynamics in $\vect x_{\texttt e} - \eqm{\vect x}_{\texttt e}[\varphi]$, which completes the proof. 
\end{proof}

Section~\ref{sec_numerical} provides numerical evidence that \eqref{E:pos_condition} and \eqref{E:diss_condition} can be verified for realistic SG parameters and that they are mildly conservative for certifying SMIB stability.

\section{Stability Implications for Multimachine Systems} \label{sec_grid}

This section completes the power network model and shows that, with the proposed rotor-angle feedback, the synchronous orbit is asymptotically stable under suitable conditions.

\subsection{Modeling Assumptions} \label{sec_assum}

The modeling assumptions for the power system are listed below (numbering continues from previous assumptions):
\begin{enumerate}
    \item[A4] \emph{Balanced three phases:} The system dynamics are invariant under permutations of the three phases. 
    \item[A5] \emph{Transmission lines:} The transmission lines are modeled by the lumped-parameter $\Pi$-model~\cite{golo2002approximation}.
    \item[A6] \emph{Shunt capacitors:} Every non-ground node is connected to the ground node by a shunt capacitor of positive capacitance. The shunt capacitance may come from the $\Pi$-model of the transmission lines and can be arbitrarily close to zero.
    \item[A7] \emph{Ground node:} The ground node has zero potential.
    \item[A8] \emph{Monotone loads:} Each load in the system is modeled by a static monotone mapping $\Upsilon: \mathbb{C}\to \mathbb{C}$, which maps from the load voltage $V$ to the load current $I$ and satisfies the following monotonicity condition~\cite{strehle2020passivity}: 
    \begin{equation*}
        \left\langle \Upsilon(V_1) - \Upsilon(V_2),\, V_1 - V_2 \right\rangle \geq 0,\quad\, \forall V_1, V_2 \in \mathbb{C}.
    \end{equation*}
\end{enumerate}

Assumption A8 is imposed to make each static load edge incrementally passive. It covers passive admittance loads and more general monotone voltage-dependent loads. Constant-power loads are not covered in general, since their incremental admittance is negative which destroys global monotonicity~\cite{monshizadeh2019power,strehle2020passivity}.

\subsection{Solvable Interconnection Constraints} \label{sec_grid_model}

We adopt the port-Hamiltonian modeling approach in~\cite{fiaz2013port}.
The interconnection of the input--output dynamical elements is modeled as a directed graph with the set of nodes $\mathcal N$ and the set of edges $\mathcal E$. The node index is $\texttt n \in \mathcal N$; the edge index is $\texttt e \in \mathcal E$. Each edge is associated with an edge voltage $V_{\textup{\texttt e}}$ and an edge current $I_{\textup{\texttt e}}$. The sign convention for $V_{\textup{\texttt e}}$ and $I_{\textup{\texttt e}}$ is consistent with the direction of the edge: $V_{\textup{\texttt e}}$ is the potential drop from ``from'' to ``to'', and $I_{\textup{\texttt e}}$ is the current through the edge from ``from'' to ``to''. The ground node, being the last node, is always assigned as the ``to'' node.
The edges carry the grid elements, which include the input--output dynamics of the SGs, the $R$--$L$ lines, and the shunt capacitors. Table~\ref{tab_1} assigns the input and output for each edge type.

\begin{table}[t]
\renewcommand{\arraystretch}{1.2}
\caption{Input and Output of Each Edge}
\label{tab_1}
\setlength\tabcolsep{2pt}
\centering
\begin{tabular}{|c|c|c|c|}
\hline
Edge type &Edge index $\textup{\texttt e}$ &Input $\vect u_{\textup{\texttt e}}$ &Output $\vect y_{\textup{\texttt e}}$ \\
\hline
SG &$1,\ldots,\, \textup{\texttt g}$ &$V_{\textup{\texttt e}}$ &$I_{\textup{\texttt e}}$ \\
%\hline
$R$--$L$ line &$\textup{\texttt g} + 1,\ldots,\, \textup{\texttt g} + \textup{\texttt T}$ &$V_{\textup{\texttt e}}$ &$I_{\textup{\texttt e}}$ \\
%\hline
Shunt &$\textup{\texttt g} + \textup{\texttt T} + 1,\ldots,\, 2 \textup{\texttt g} + \textup{\texttt T} + \ell $ &$I_{\textup{\texttt e}}$ &$V_{\textup{\texttt e}}$ \\
\hline
\end{tabular}
\end{table}

From Assumption A5, the $\Pi$-model of a transmission line section consists of one $R$--$L$ line and two shunt capacitors. Table~\ref{tab_1} shows that there are $\textup{\texttt g}$ SGs, $\textup{\texttt T}$ $R$--$L$ lines, and $\textup{\texttt g} + \ell$ shunt capacitors with loads. By Assumption A6, there are $\textup{\texttt g}$ nodes that are connected to the ground by an SG and a shunt capacitor in parallel. There are $\ell$ nodes that are connected to the ground by a shunt capacitor only. The incidence matrix for the graph of the network can be written as
\begin{equation} \label{E:incidence}
    \matr N = \begin{bmatrix}
        \begin{bmatrix}
            \matr I_{\textup{\texttt g}} \\
            \matr 0_{\ell\times \textup{\texttt g}}
        \end{bmatrix} &\matr N_1 &\matr I_{\textup{\texttt g} + \ell}  \\[10pt]
        -\mathbbm{1}_{\textup{\texttt g}}^\tran &\matr 0_{\textup{\texttt T}}^\tran &-\mathbbm{1}_{\textup{\texttt g} + \ell}^\tran 
    \end{bmatrix},
\end{equation}
where $\matr N_1 \in \{-1, 0, 1\}^{(|\mathcal N| - 1) \times \texttt T}$ is the sub-incidence matrix of the subgraph obtained by removing the ground node~\cite{fiaz2013port}.

The edge voltages and currents are physically constrained by KVL and KCL according to the graph incidence. These constraints are written as
\begin{align*}
    &\matr N\, \col(I_{\textup{\texttt e}})_{\textup{\texttt e}=1}^{|\mathcal E|} = \matr 0_{|\mathcal N|}, \tag{KCL} \\
    &\matr N^\tran\, \col(E_{\texttt n})_{\texttt n=1}^{|\mathcal N|} = \col(V_{\textup{\texttt e}})_{\textup{\texttt e}=1}^{|\mathcal E|}, \tag{KVL}
\end{align*}
where $E_\texttt n$, $\texttt n=1,\ldots, |\mathcal N|$ are the nodal voltage potentials.

With the edge ordering in Table~\ref{tab_1}, KVL determines the voltage
inputs of the SG and $R$--$L$ line edges from the shunt-voltage outputs,
whereas KCL determines the current inputs of the shunt edges from
the SG and line-current outputs. Since every non-ground node has a
shunt capacitor by Assumption A6, these node voltages are state/output
variables. Therefore, the KCL and KVL constraints can be solved
explicitly as the linear mapping $\vect y \mapsto \vect u = \matr W \vect y$,
where $\vect y = \col(\vect y_{\textup{\texttt e}})_{\textup{\texttt e}=1}^{|\mathcal E|}$, $\vect u = \col(\vect u_{\textup{\texttt e}})_{\textup{\texttt e}=1}^{|\mathcal E|}$, and
\begin{equation} \label{E:W}
    \matr W = \begin{bmatrix}
        \matr 0_{\textup{\texttt g}\times \textup{\texttt g}} &\matr 0_{\textup{\texttt g}\times \textup{\texttt T}} &\begin{bmatrix}
            \matr I_{\textup{\texttt g}} &\matr 0_{\textup{\texttt g}\times \ell}
        \end{bmatrix} \\[5pt]
        \matr 0_{\textup{\texttt T}\times \textup{\texttt g}} &\matr 0_{\textup{\texttt T}\times \textup{\texttt T}} &\matr N_1^\tran \\[5pt]
        \begin{bmatrix}
            -\matr I_{\textup{\texttt g}} \\
            \matr 0_{\ell\times \textup{\texttt g}}
        \end{bmatrix} &-\matr N_1 &\matr 0_{\textup{\texttt g} + \ell}
    \end{bmatrix}.
\end{equation}
Since $\matr W = -\matr W^\tran = -\matr W^*$, we have
\begin{align}
    &\quad\, \sum_{\texttt e=1}^{|\mathcal E|} \left\langle \vect y_{\texttt e} - \eqm{\vect y}_{\texttt e}[\varphi], \vect u_{\texttt e} - \eqm{\vect u}_{\texttt e}[\varphi] \right\rangle = \left\langle \vect y - \eqm{\vect y}[\varphi], \vect u - \eqm{\vect u}[\varphi] \right\rangle 
    \notag \\
    &= \left\langle \vect y - \eqm{\vect y}[\varphi], \matr W(\vect y - \eqm{\vect y}[\varphi]) \right\rangle = 0. \label{E:cancellation}
\end{align}
That is, the sum of the supply rates of every edge cancels.

\begin{remark}
We eliminated the algebraic variables $E_\texttt n,\, \texttt n=1,\ldots, |\mathcal N|$, that are not associated with energy storage, but are necessary to write down KVL. This elimination resembles the time-domain Kron reduction~\cite{caliskan2014towards,singh2022time,kettner2017properties}. The feasibility condition here is embedded in the feedback inputs and outputs assigned for each edge, without requiring an additional algebraic condition.
\end{remark}

\subsection{Energy Balance for Each Edge}

\subsubsection{Synchronous Generator}

The augmented generator storage
$\widetilde H_{\texttt e}(\vect x_{\texttt e}, \varphi; \sigma_{\texttt e})$ is locally positive definite with respect to $\vect x_{\texttt e} = \eqm{\vect x}_{\texttt e}[\varphi]$ by Proposition~\ref{lem_pos}, and the
second-order part of its time derivative is given by Proposition~\ref{lem_diss}.

\subsubsection{R--L Line} \label{sec_RL}
The equation for each $R$--$L$ line is
\begin{equation} \label{E:R-L}
    L_\textup{\texttt e} \dot I_\textup{\texttt e} = -R_\textup{\texttt e} I_\textup{\texttt e} + V_\textup{\texttt e},
\end{equation}
where $R_{\textup{\texttt e}}> 0$ and $ L_{\textup{\texttt e}} > 0$ are the resistance and inductance. Let $\vect x_{\texttt e} \coloneqq I_{\texttt e}$ and 
$H_{\texttt e}(\vect x_{\texttt e}, \varphi) \coloneqq \frac{1}{2} L_{\texttt e} \|I_{\texttt e} - \eqm I_{\texttt e}[\varphi] \|^2$.
We can check that along the flow of the extended dynamics consisting of \eqref{E:R-L} and $\dot\varphi = \eqm\omega$,
\begin{align}
    \dot H_{\texttt e} (\vect x_{\texttt e}, \varphi) 
    &= -R_{\texttt e}\|I_{\texttt e} - \eqm I_{\texttt e}[\varphi]\|^2 \notag \\
    &\quad\, + \left\langle \vect y_{\texttt e} - \eqm{\vect y}_{\texttt e}[\varphi], \vect u_{\texttt e} - \eqm{\vect u}_{\texttt e}[\varphi] \right\rangle. \label{E:pass2}
\end{align}

\subsubsection{Shunt Capacitor in Parallel with a Monotone Load and Rotor-Angle Feedback} \label{sec_shunt}

The equation for each shunt capacitor with load and injected currents is
\begin{equation} \label{E:shunt}
    C_\textup{\texttt e} \dot V_\textup{\texttt e} = -G_{\textup{\texttt e}} V_{\textup{\texttt e}} + \gamma_{\texttt e} - \Upsilon_\textup{\texttt e}(V_\textup{\texttt e}) + I_\textup{\texttt e},
\end{equation}
where $C_{\textup{\texttt e}}> 0$ and $G_{\textup{\texttt e}} > 0$ are the capacitance and conductance, and $\Upsilon_{\textup{\texttt e}}(V_{\textup{\texttt e}})$ is the current drawn by the static load. The proposed rotor-angle feedback $\gamma_{\texttt e}$ is defined as
\begin{align}
    &\gamma_{\texttt e} \coloneqq \notag \\
    &\begin{cases}
        \displaystyle -\frac{\lambda_{\texttt i}}{L_{\texttt i}} \left(\xi_{\texttt i} - \frac{\eqm\xi_{\texttt i}[0]}{\eqm\xi_1[0]} \xi_1 \right), &\texttt e= \texttt i + \texttt g+\texttt T,\, \texttt i = 1,\ldots, \texttt g, \\
        0, &\texttt e = 2\texttt g+\texttt T +1,\ldots, 2\texttt g+ \texttt T +\ell.
    \end{cases} \label{E:feedback}
\end{align}
By the incidence matrix $\matr N$ in \eqref{E:incidence}, $\xi_{\texttt i}$ is the rotor angle of the SG adjacent to the shunt capacitor, and $\xi_1$ is the rotor angle of the first SG. In the steady state, we can verify that $\gamma_{\texttt e} = 0$. Hence $\gamma_{\texttt e}$ does not change the steady state, cf. Assumption~\ref{assum_exist}.

\begin{remark}
The rotor-angle feedback $\gamma_{\texttt e}$ resembles consensus algorithms in the synchronization of multi-agent systems. To implement $\gamma_{\texttt e}$, it requires knowledge of $\eqm\xi_{\texttt i}[0]/\eqm\xi_1[0]$, i.e., the steady-state rotor angle difference between every SG in the system. Approximate steady-state rotor angle differences are typically obtained from power flow calculation. 
For the sake of the analysis, we assume that the precise steady-state rotor angle differences are known, cf. \cite{colombino2019global,gross2019effect}. 
\end{remark}

Let $\vect x_{\texttt e} \coloneqq V_{\texttt e}$ and $H_{\texttt e}(\vect x_{\texttt e}, \varphi) \coloneqq \frac{1}{2} C_{\texttt e}\|V_{\texttt e} - \eqm V_{\texttt e}[\varphi]\|^2$.
Along the flow of the extended dynamics consisting of \eqref{E:shunt} and $\dot\varphi = \eqm\omega$, we have, for $\texttt e = \texttt i + \texttt g+\texttt T$ with $\texttt i = 1,\ldots,\texttt g$,
\begin{align}
    &\quad\,\, \dot H_{\texttt e} (\vect x_{\texttt e}, \varphi) \notag \\
    &= -G_{\texttt e} \|V_{\texttt e} - \eqm V_{\texttt e}[\varphi]\|^2 - \frac{\lambda_{\texttt i}}{L_{\texttt i}} \left\langle V_{\texttt e} - \eqm V_{\texttt e}[\varphi], \xi_{\texttt i} - \eqm\xi_{\texttt i}[\varphi] \right\rangle \notag \\
    &\quad\, + \frac{\lambda_{\texttt i}}{L_{\texttt i}} \left\langle V_{\texttt e} - \eqm V_{\texttt e}[\varphi], \frac{\eqm\xi_{\texttt i}[0]}{\eqm \xi_1[0]} (\xi_1 - \eqm\xi_1[\varphi])\right\rangle \notag \\
    &\quad\, + \left\langle \vect y_{\texttt e} - \eqm{\vect y}_{\texttt e}[\varphi], \vect u_{\texttt e} - \eqm{\vect u}_{\texttt e}[\varphi] \right\rangle \notag \\
    &\quad\, - \left\langle V_{\texttt e} - \eqm V_{\texttt e}[\varphi], \Upsilon_{\texttt e}(V_{\texttt e}) - \Upsilon_{\texttt e}(\eqm V_{\texttt e}[\varphi] ) \right\rangle \notag \\
    &\leq -G_{\texttt e} \|V_{\texttt e} - \eqm V_{\texttt e}[\varphi]\|^2 - \frac{\lambda_{\texttt i}}{L_{\texttt i}} \left\langle \vect y_{\texttt e} - \eqm{\vect y}_{\texttt e}[\varphi], \xi_{\texttt i} - \eqm\xi_{\texttt i}[\varphi] \right\rangle \notag \\
    &\quad\, + \frac{\lambda_{\texttt i}}{L_{\texttt i}} \left\langle V_{\texttt e} - \eqm V_{\texttt e}[\varphi], \frac{\eqm\xi_{\texttt i}[0]}{\eqm \xi_1[0]} (\xi_1 - \eqm\xi_1[\varphi])\right\rangle \notag \\
    &\quad\, + \left\langle \vect y_{\texttt e} - \eqm{\vect y}_{\texttt e}[\varphi], \vect u_{\texttt e} - \eqm{\vect u}_{\texttt e}[\varphi] \right\rangle. \label{E:pass3}
\end{align}
where we used Assumption A8 to remove the load contribution, and, for $\texttt e = 2\texttt g+\texttt T+1,\ldots, 2\texttt g+\texttt T+\ell$,
\begin{align}
    \dot H_{\texttt e} (\vect x_{\texttt e}, \varphi)
    &\leq -G_{\texttt e}\|V_{\texttt e} - \eqm V_{\texttt e}[\varphi]\|^2 \notag \\
    &\quad\, + \left\langle \vect y_{\texttt e} - \eqm{\vect y}_{\texttt e}[\varphi], \vect u_{\texttt e} - \eqm{\vect u}_{\texttt e}[\varphi] \right\rangle. \label{E:pass4}
\end{align}

\subsection{Network Stability Result}

\begin{proposition} \label{prop_main}
Consider the power system model \eqref{E:system} with the rotor-angle feedback \eqref{E:feedback}. Assume Assumptions 1--3, A4--A8, and assume that, for each SG with index $\texttt e =1,\ldots, \texttt g$, there exists $\sigma_{\texttt e} > 0$ that verifies the local shifted dissipativity condition \eqref{E:diss_condition}. In addition, assume
\begin{equation}
\begin{aligned}
&
\frac{1}{
d_1
}
-
\sum_{\texttt i=1}^{\texttt g}
\frac{(\lambda_{\texttt i}/L_{\texttt i})^2}{4 G_{\texttt i+\texttt g+\texttt T}}
+
\frac{\left(n_1/
d_1
-
q
\right)^2
}{
\displaystyle
\sum_{\texttt e=1}^{\texttt g}
\vect b_{\texttt e}^\tran \matr M_{\texttt e}(\sigma_{\texttt e})^{-1}\vect b_{\texttt e}
-
\frac{
n_1^2}{
d_1
} + \tau
}
>0,
\end{aligned}
\label{eq:phase-dynamics-condition}
\end{equation}
where
\begin{align*} 
&q\coloneqq \sum_{\texttt i=1}^{\texttt g}
\frac{\lambda_{\texttt i} C_{\texttt i+\texttt g+\texttt T}}
     {2L_{\texttt i}G_{\texttt i+\texttt g+\texttt T}}
\left\langle\bar\xi_{\texttt i}[0],\bar V_{\texttt i+\texttt g+\texttt T}[0]\right\rangle , \notag \\
&\tau \coloneqq \sum_{\texttt e=1}^{\texttt g}
\frac{
L_{\texttt e}^2
\left\langle\bar I_{\texttt e}[0],j\bar\xi_{\texttt e}[0] \right\rangle^2
}{R_{\texttt e}}
+
\sum_{\texttt e=\texttt g+1}^{\texttt g+\texttt T}
\frac{L_{\texttt e}^2\|\bar I_{\texttt e}[0]\|^2}{R_{\texttt e}} \notag \\
&\qquad\:\, +
\sum_{\texttt e=\texttt g+\texttt T+1}^{2\texttt g+\texttt T+\ell}
\frac{C_{\texttt e}^2\|\bar V_{\texttt e}[0]\|^2}{G_{\texttt e}}, \notag \\
&d_{\texttt e}\coloneqq\vect e_2^\tran \matr M_{\texttt e}(\sigma_{\texttt e})^{-1} \vect e_2,
\quad
n_{\texttt e}\coloneqq\vect e_2^\tran \matr M_{\texttt e}(\sigma_{\texttt e})^{-1} \vect b_{\texttt e},
\\
&\vect e_2\coloneqq \begin{bmatrix}
0 &1 &0 \end{bmatrix}^\tran,\quad \vect b_{\texttt e}
\coloneqq \notag \\
&
\begin{bmatrix}
-\sigma_{\texttt e}J_{\texttt e} &
-\displaystyle \left(
\sigma_{\texttt e}D_{\texttt e}+2\mathcal K_{\texttt e}+\frac{\lambda_{\texttt e}^2}{L_{\texttt e}}
\right) &
\displaystyle \frac{L_{\texttt e}}{\lambda_{\texttt e}}\mathcal K_{\texttt e}+\lambda_{\texttt e}
\end{bmatrix}^\tran. 
\end{align*}
Then the synchronous orbit \(\Gamma\) is asymptotically stable.
\end{proposition}

\begin{proof}
Define the error coordinates
\begin{align*}
&\delta_{\texttt e}
\coloneqq
\arg(\eqm\xi_{\texttt e}[\varphi]^*\xi_{\texttt e}), &&\texttt e= 1,\ldots,\texttt g, \\
&\nu_{\texttt e}\coloneqq\omega_{\texttt e}-\eqm\omega, &&\texttt e = 1,\ldots,\texttt g, \\
&\zeta_{\texttt e} \coloneqq I_{\texttt e}-\eqm I_{\texttt e}[\varphi], &&\texttt e = 1,\ldots,\texttt g+\texttt T, \\
&\eta_{\texttt e}\coloneqq V_{\texttt e}-\eqm V_{\texttt e}[\varphi], &&\texttt e = \texttt g+\texttt T+1,\ldots,2\texttt g+\texttt T+\ell,
\end{align*}
where the branch of \(\arg(\cdot)\) is chosen locally near zero.
In vector form, let
$\boldsymbol{\delta} = \col(\delta_{\texttt e})_{\texttt e=1}^{\texttt g}, \, \boldsymbol{\nu} = \col(\nu_{\texttt e})_{\texttt e=1}^{\texttt g},\, \boldsymbol{\zeta} = \col(\zeta_{\texttt e})_{\texttt e=1}^{\texttt g+\texttt T},\, \boldsymbol{\eta} = \col(\eta_{\texttt e})_{\texttt e=\texttt g+\texttt T+1}^{2\texttt g+\texttt T+\ell}
$, and 
$\vect z \coloneqq \col(\boldsymbol{\nu},\boldsymbol{\delta},\boldsymbol{\zeta},\boldsymbol{\eta})$.
Define the extended orbit
\begin{equation*}
    \Gamma_{\mathrm{ext}} \coloneqq \left\{(\eqm{\vect x}[\varphi], \varphi)\mid \varphi\in \mathbb{T} \right\}.
\end{equation*}

Define the aggregated network storage
\begin{align} \label{E:aggregate}
\widetilde H(\vect x, \varphi)
&\coloneqq
\sum_{\texttt e=1}^{\texttt g} \widetilde H_{\texttt e}(\vect x_{\texttt e},\varphi; \sigma_{\texttt e}) + \sum_{\texttt e=\texttt g+1}^{2\texttt g+\texttt T+\ell} H_{\texttt e}(\vect x_{\texttt e},\varphi).
\end{align}
By Proposition~\ref{lem_pos}, Remark~\ref{rem_pos}, $L_{\texttt e} > 0$, and $C_{\texttt e} > 0$, we have that $\widetilde H(\vect x, \varphi)$ is locally positive definite with respect to $\Gamma_{\mathrm{ext}}$ in a small neighborhood $U_0 \subset \mathbb{X}\times \mathbb{T}$.
Combining the edge balance equations \eqref{E:eb_raw2}, \eqref{E:pass2}, \eqref{E:pass3}, and \eqref{E:pass4}, and using the interconnection constraint \eqref{E:cancellation}, we obtain, along the flow of the extended dynamics consisting of \eqref{E:system} and $\dot\varphi = \eqm\omega$,
\begin{align}
&\quad\,\, \dot{\widetilde H}(\vect x, \varphi) \notag \\
&\leq
-\widetilde{\mathcal Q}(\vect z)
+ \sum_{\texttt i=1}^{\texttt g} \frac{\lambda_{\texttt i}}{L_{\texttt i}} \delta_1 \left\langle\eta_{\texttt i + \texttt g+\texttt T}, j \eqm\xi_{\texttt i}[\varphi] \right\rangle + O\!\left(\|\vect z\|^3 \right),
\label{E:Htilde-dot}
\end{align}
where
\begin{align}
&\quad\,\, \widetilde{\mathcal Q}(\vect z) \notag \\
&=
\sum_{\texttt e=1}^{\texttt g}
\widetilde{\mathcal Q}_{\texttt e}(\nu_{\texttt e}, \delta_{\texttt e}, \zeta_{\texttt e}) +
\sum_{\texttt e=\texttt g+1}^{\texttt g+\texttt T}R_{\texttt e}\|\zeta_{\texttt e}\|^2 + \sum_{\texttt e=\texttt g+\texttt T+1}^{2\texttt g+\texttt T+\ell} G_{\texttt e} \|\eta_{\texttt e}\|^2. \notag %\label{E:Qtilde2-expanded}
\end{align}
By Proposition~\ref{lem_diss}, \(R_{\texttt e}>0\), and \(G_{\texttt e}>0\), the quadratic form $\widetilde{\mathcal Q}(\vect z)$ is positive definite. 

Now consider the phase dynamics, for some $c > 0$,
\begin{equation} \label{E:gauge}
    \dot{\hat\varphi} = \eqm\omega - c \frac{\partial \widetilde H}{\partial \hat \varphi} (\vect x, \hat\varphi),
\end{equation}
which is to replace the original phase dynamics $\dot\varphi = \eqm\omega$. Along the flow of the extended dynamics consisting of \eqref{E:system} and \eqref{E:gauge},
\begin{align}
    &\quad\,\, \dot{\widetilde H}(\vect x, \hat\varphi) \notag \\
    &= \left.\dot{\widetilde H}(\vect x, \varphi) \right|_{(\vect x, \varphi) = (\vect x, \hat\varphi)} + \frac{\partial \widetilde H}{\partial \hat \varphi} (\vect x, \hat\varphi) (\dot{\hat\varphi} - \dot\varphi)\notag \\
    &\leq -\widetilde{\mathcal Q}(\hat{\vect z}) + \sum_{\texttt i=1}^{\texttt g} \frac{\lambda_{\texttt i}}{L_{\texttt i}} \hat\delta_1 \left\langle \hat\eta_{\texttt i + \texttt g+\texttt T}, j \eqm\xi_{\texttt i}[\hat\varphi] \right\rangle - c \left[ \frac{\partial \widetilde H}{\partial \hat \varphi} (\vect x, \hat\varphi) \right]^2 \notag \\
    &\quad\, + O\!\left(\|\hat{\vect z}\|^3 \right), \label{E:last_balance}
\end{align}
where the error coordinates $\hat{\vect z}$ are referenced to $\hat\varphi$ which has the dynamics \eqref{E:gauge}.

\begin{lemma}
\label{lem:phase-dynamics}
Suppose \(\matr M_{\texttt e}(\sigma_{\texttt e})\succ0\) for every \(\texttt e=1,\ldots, \texttt g\). 
Then, there exists a finite \(c>0\) such that the quadratic part of
\eqref{E:last_balance} is negative definite if and only if \eqref{eq:phase-dynamics-condition} holds.
\end{lemma}

The proof of Lemma~\ref{lem:phase-dynamics} is given in Appendix~\ref{proof_phase_dynamics}.

By Lemma~\ref{lem:phase-dynamics}, there exists a small neighborhood $U_1 \subset U_0$ of $\Gamma_{\mathrm{ext}}$ in which the negative-definite quadratic part of \eqref{E:last_balance} dominates the cubic remainder. Hence $\dot{\widetilde H}(\vect x, \hat\varphi)$ is negative definite with respect to $\Gamma_{\mathrm{ext}}$.
By~\cite[Theorem~2.7.1]{bhatia2006dynamical}, we conclude that $\Gamma_{\mathrm{ext}}$ is asymptotically stable. This implies that $\Gamma$ is asymptotically stable.
\end{proof}

\section{Numerical Examples}
\label{sec_numerical}

We illustrate the feasibility of the shifted dissipativity
conditions
\eqref{E:diss_condition} with the constraint $\sigma > 0$ and verify  Corollary~\ref{cor} for SMIB stability.
The SG parameters are adapted from
\cite{anderson2004power}.
The nominal parameters
are listed in Table~\ref{tab:sync-params}. For each value of a
swept parameter, the two branches of synchronous steady states are
computed from the steady-state generator equations. Condition
\eqref{E:diss_condition} is then checked separately on each
branch. 
For a branch satisfying~\eqref{E:diss_condition}, define the
certificate margin
\begin{equation}
    m\coloneqq 
    \max_{\sigma > 0}
    \lambda_{\min}\!\left(\matr M(\sigma)\right).
    \label{eq:certificate-margin}
\end{equation}
Thus, \(m>0\) implies that there exists a value of \(\sigma\) for
which \(\matr M(\sigma)\succ0\), and hence that the branch admits a valid
local shifted dissipativity certificate.

\begin{table}[!t]
\centering
\caption{Synchronous Machine Parameters}
\label{tab:sync-params}
\begin{tabular}{|l|c|l|c|}
\hline
Parameter & Value [p.u.] & Parameter & Value [p.u.] \\
\hline
$\eqm\omega$ & $1$ & $J$ & $3.53$ \\
$R$ & $0.0211$ & $\lambda$ & $0.711$ \\
$L$ & $2.1$ & $\eqm V$ & $1$ \\
$D$ & $19$ &$\eqm P$ &$-0.2$ \\
%Reactive power $\eqm Q$ & $0.15$ \\
\hline
\end{tabular}
\end{table}

\begin{figure*}[!t]
\centerline{

\includegraphics[width=0.672\columnwidth,valign=b]{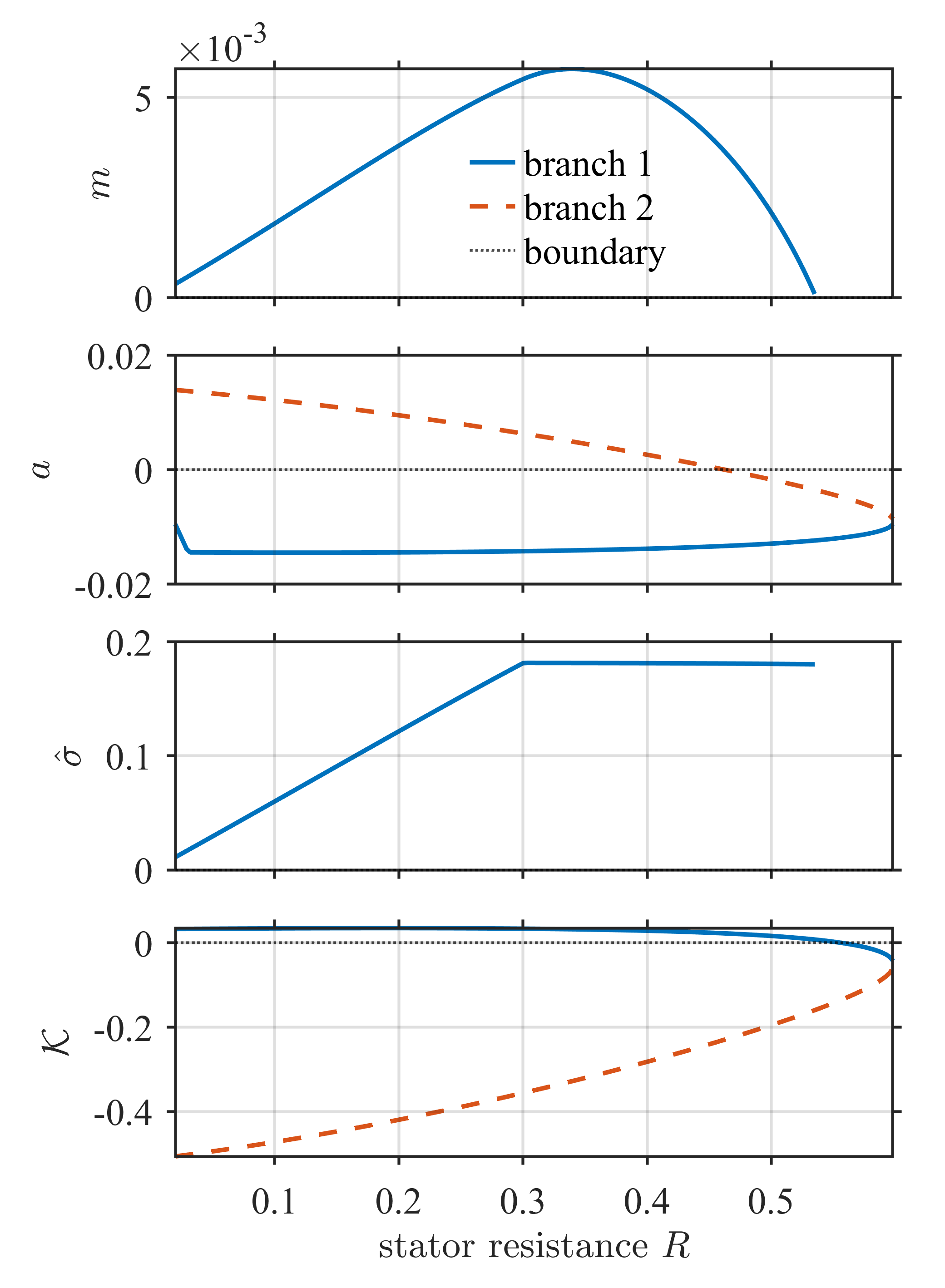}
\includegraphics[width=0.659\columnwidth,valign=b]{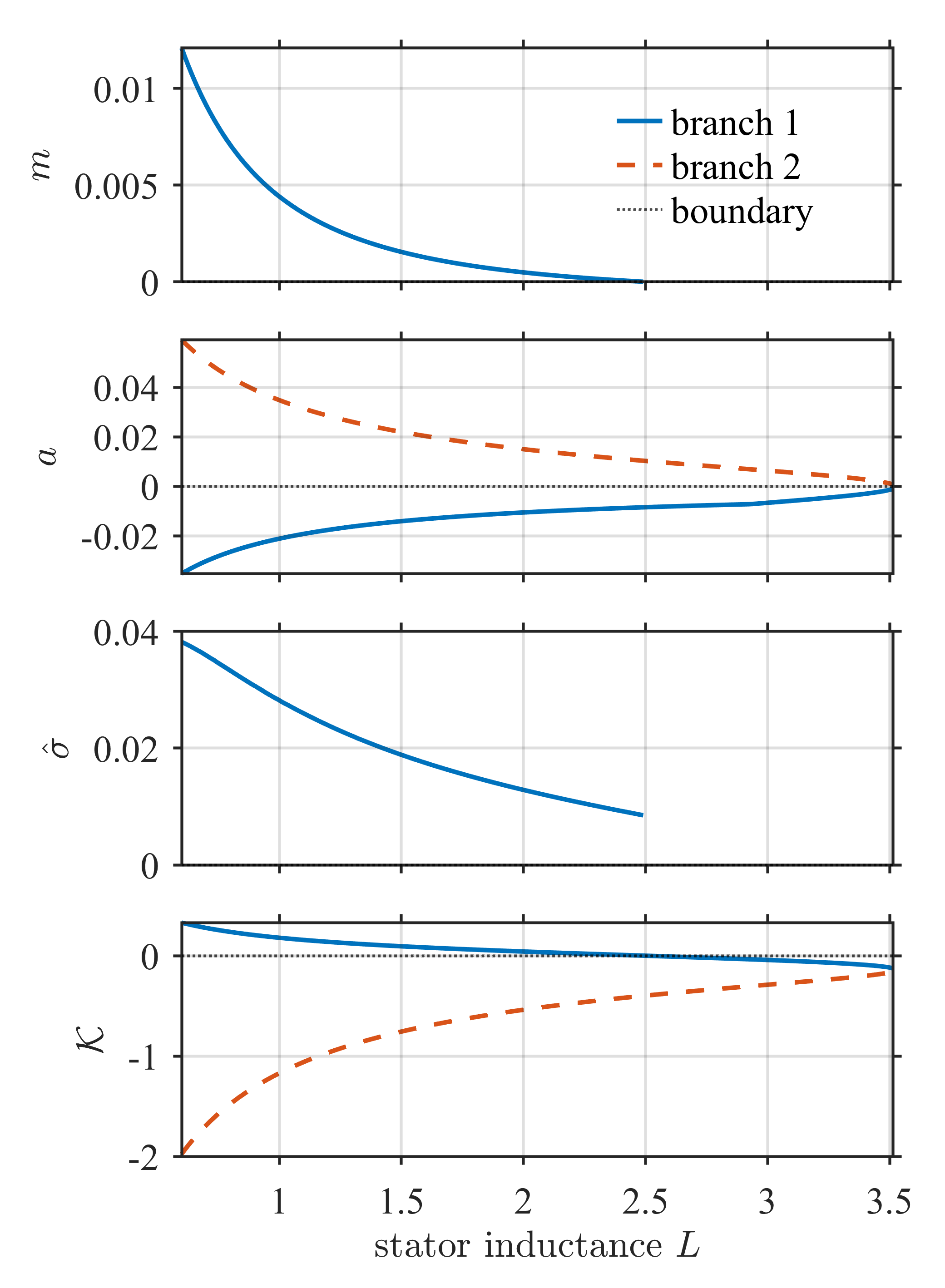}
\includegraphics[width=0.67\columnwidth,valign=b]{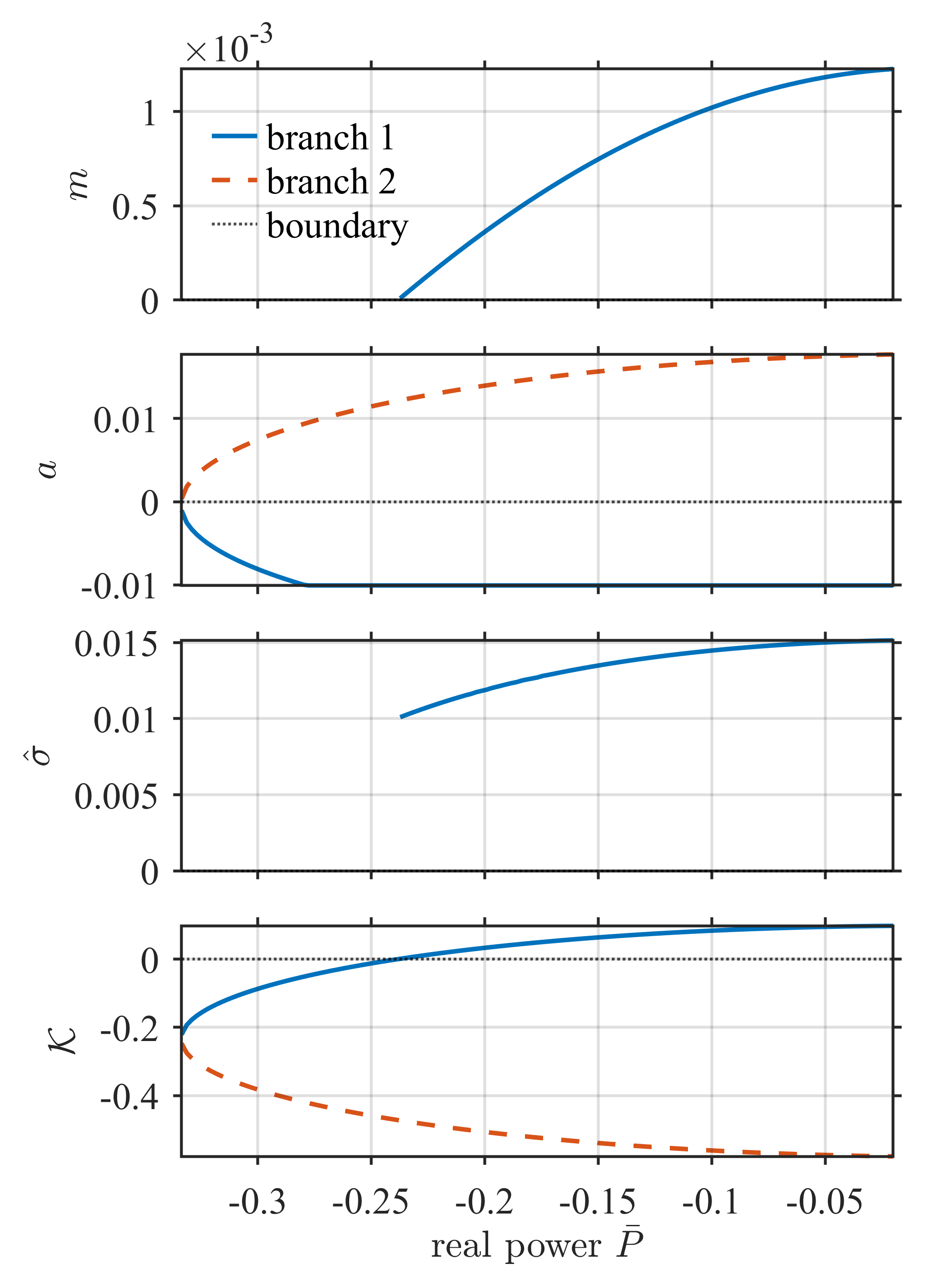}
}
\caption{Feasibility study for the proposed shifted dissipativity conditions
    under one-parameter sweeps of {\unboldmath \(R\)}, {\unboldmath \(L\)}, and {\unboldmath \(\bar P\)}.
    Each column corresponds to one swept parameter. From top to
    bottom, the rows show the optimized certificate margin {\unboldmath \(m\)},
    the spectral abscissa of the SMIB linearization {\unboldmath \(a\)}, the
    certifying optimizer {\unboldmath \(\hat\sigma\)}, and
    {\unboldmath \(\mathcal K=\lambda\langle
    \bar I[\varphi],\bar\xi[\varphi]\rangle\)}.
    Certificate margin and optimizer curves are shown only where
    {\unboldmath \(m>0\)}.}
\label{fig:certificate_sweeps}
\end{figure*}

To validate the dissipativity certificate, we linearize the SMIB system in a reference frame rotating at the frequency $\eqm\omega$ at the two branches of synchronous steady states. The spectral abscissa
\begin{equation}
    a \coloneqq \max_{\lambda_i\in\operatorname{spec}(\matr A_{\mathrm{inf}})}
    \operatorname{Re}\lambda_i
\end{equation}
is computed, where \(\matr A_{\mathrm{inf}}\) is the corresponding linearization matrix. Hence \(a<0\) indicates
local exponential stability of the linearized SMIB model.

Fig.~\ref{fig:certificate_sweeps} shows one-parameter sweeps in
the stator resistance \(R\), stator inductance \(L\), and
steady-state real power \(\bar P\). Each column corresponds to one
swept parameter. The first row shows the optimized certificate
margin \(m\), the second row shows the spectral abscissa
\(a\), the third row shows the certifying optimizer
\(\hat\sigma\), and the fourth row shows $\mathcal K$.
The certificate-margin and optimizer curves are displayed only at
parameter values for which \(m>0\). Across the displayed sweeps, the certificate does not certify the unstable branch and identifies a nontrivial subset of the stable branch. Although multimachine simulations would provide additional illustration, they are not essential to substantiate the analytical result and are omitted owing to the space limitation.
\section{Conclusion} \label{sec_concl}

This paper introduced a shifted dissipativity condition for SGs tailored to the orbital-stability analysis of multimachine power networks. The condition guarantees local asymptotic stability in the SMIB setting and, together with a consensus-like rotor-angle feedback, local orbital stability of a multimachine power-network model that retains transmission-line dynamics. The approach addresses rotor-angle alignment by using an internal-model argument to represent the residual phase dynamics of the synchronous orbit. A potential application is the design of consensus feedback aimed at improving orbital-stability margins.

\section*{Acknowledgment}

The author used ChatGPT-5.6 Sol to revise the proofs of Lemmas~\ref{lem_equiv} and~\ref{lem:phase-dynamics} and to assist in creating the code for Fig.~\ref{fig:certificate_sweeps}.
The author independently
verified these results. The author is grateful to his doctoral committee members at the Pennsylvania State University and anonymous reviewers for previous incarnations of this work, for their valuable comments.
\appendix

\subsection{Proof of Lemma~\ref{lem_equiv}} \label{proof_equiv}
\color{black}

\begin{proof}
\emph{Step 1 (Definition~\ref{defn_2} $\Rightarrow$ Definition~\ref{defn_1}).} Let $b_0 \coloneqq \inf_{(\vect x_{\texttt e}, \varphi)\in U} S(\vect x_{\texttt e}, \varphi)$.
Then $S(\vect x_{\texttt e}(t), \varphi(t)) \geq b_0$ for all $t \geq 0$. This gives $\int_0^{t} s(\vect u_{\texttt e}(\tau), \vect y_{\texttt e}(\tau), \varphi(\tau))\, d \tau \geq S(\vect x_{\texttt e}(t), \varphi(t)) - S(\vect x_{\texttt e}(0), \varphi(0)) \geq b_0 - S(\vect x_{\texttt e}(0), \varphi(0))$. Hence \eqref{E:shifted_passivity} holds with $b(\vect x_{\texttt e}(0), \varphi(0)) \coloneqq S(\vect x_{\texttt e}(0), \varphi(0)) - b_0 \geq 0$. 

\emph{Step 2 (Definition~\ref{defn_2} $\Leftarrow$ Definition~\ref{defn_1}).} We prove this constructively by considering the candidate storage function 
\begin{align*}
    &S(\vect x_{\texttt e}(0), \varphi(0)) \coloneqq {} \notag \\
    & \sup_{\substack{t\in \mathbb{R}_{\geq 0}, \\ (\vect u_{\texttt e}(\cdot) - \eqm{\vect u}_{\texttt e}[\varphi(\cdot)]) \in \mathbf{L}^2_{\mathrm{loc}}([0,\, \infty);\, \mathbb{C}), \\ (\vect x_{\texttt e}(\tau), \varphi(\tau)) \in U, \, \tau \in [0,\, t]}} {-}\int_0^{t} s(\vect u_{\texttt e}(\tau), \vect y_{\texttt e}(\tau), \varphi(\tau))\, d \tau
\end{align*}
where the trajectories over which the supremum is taken are initialized at $(\vect x_{\texttt e}(0), \varphi(0))$ and remain in $U$. 
Note that \eqref{E:shifted_passivity} implies $-\int_0^{t} s(\vect u_{\texttt e}(\tau), \vect y_{\texttt e}(\tau), \varphi(\tau))\, d \tau \leq b(\vect x_{\texttt e}(0), \varphi(0))$. Hence $S(\vect x_{\texttt e}(0), \varphi(0)) \leq b(\vect x_{\texttt e}(0), \varphi(0))$ is finite.
Since $S(\vect x_{\texttt e}, \varphi)$ is defined as a supremum over $t \in \mathbb{R}_{\geq 0}$ and $t = 0$ is admissible, we have $S(\vect x_{\texttt e}, \varphi)\geq 0$. Finally, by the fundamental theorem of dynamic programming, $S(\vect x_{\texttt e}(0), \varphi(0)) \geq -\int_0^t s(\vect u_{\texttt e}(\tau), \vect y_{\texttt e}(\tau), \varphi(\tau))\, d \tau + S(\vect x_{\texttt e}(t), \varphi(t))$, which rearranges to \eqref{E:shifted_passivity2}.
\end{proof}

\subsection{Proof of Proposition~\ref{prop_ineq}} \label{proof_ineq}

\begin{proof}
\emph{Step 1.} We first cast  \eqref{E:swing}, \eqref{E:stator}, and $\dot\varphi = \eqm\omega$ as an almost port-Hamiltonian error system. We rewrite \eqref{E:swing} and \eqref{E:stator} as
\begin{subequations} \label{E:swing3} \begin{align}
    &\frac{d }{d t} (j J \omega\xi) 
    = -D j \omega\xi + j T_m \xi 
    + \frac{\lambda}{2} (I - I^*\xi^2) - J \eqm\omega^2 \xi, \\
    &\dot\xi = j\omega\xi \\
    &L \dot I = -R I - j \lambda\omega\xi + V.
\end{align} \end{subequations}
Equations for a steady-state trajectory $\eqm{\vect x}_{\texttt e}[\varphi]$ with $\dot\varphi = \eqm\omega$ are written similarly.
Let us change to the error coordinates
\begin{align*}
    &\vect z_{\texttt e} \coloneqq \begin{bmatrix}
        j J \omega \xi - j J \eqm\omega \eqm\xi[\varphi] &
        \xi - \eqm\xi[\varphi] &
        LI - L \eqm I[\varphi]
    \end{bmatrix}^\tran.
\end{align*}
Define
\begin{align*}
    \mathcal H_{\texttt e}(\vect z_{\texttt e}) &\coloneqq H_{\texttt e}(\vect x_{\texttt e}, \varphi) - \frac{\lambda^2}{2L} \|\xi - \eqm\xi[\varphi]\|^2 \\
    &\quad\,\, - \lambda \left\langle I - \eqm I[\varphi], \xi - \eqm\xi[\varphi] \right\rangle,
\end{align*}
which is a quadratic function of $\vect z_{\texttt e}$. 
We have
\begin{align}
    &\nabla \mathcal H_{\texttt e}(\vect z_{\texttt e}) = \begin{bmatrix}
        j\omega\xi - j\eqm\omega\eqm\xi[\varphi] \\[1pt]
        \lambda  \left\langle \eqm I[\varphi], \eqm\xi[\varphi] \right\rangle (\xi - \eqm\xi[\varphi]) \\[1pt]
        I - \eqm I[\varphi]
    \end{bmatrix} \eqqcolon \begin{bmatrix}
        g_1 \\
        g_2 \\
        g_3
    \end{bmatrix}. \label{E:gradient_H}
\end{align}
From \eqref{E:swing1}, the steady-state torque balance is
\begin{equation} \label{E:ss_swing}
    T_m = D\eqm\omega - \lambda \left\langle \eqm I[\varphi], j \eqm\xi[\varphi] \right\rangle.
\end{equation}
Using \eqref{E:ss_swing} and
\begin{equation*}
    -j \lambda \left\langle \eqm I[\varphi], j \eqm\xi[\varphi]\right\rangle + \lambda \left\langle \eqm I[\varphi], \eqm\xi[\varphi] \right\rangle = \lambda \eqm I[\varphi]^* \eqm\xi[\varphi],
\end{equation*}
we cast \eqref{E:swing3} and \eqref{E:gradient_H} into the almost port-Hamiltonian form:
\begin{align}
    &\dot{\vect z}_{\texttt e} = \matr F \nabla \mathcal H_{\texttt e}(\vect z_{\texttt e}) \notag \\
    &+ \begin{bmatrix}
        -J \omega^2 \xi + J \eqm\omega^2\eqm\xi[\varphi] + 
\left(j D\eqm\omega + \lambda \eqm I[\varphi]^* \eqm\xi[\varphi]\right) (\xi - \eqm\xi[\varphi]) \\
        0 \\
        V - \eqm V[\varphi]
    \end{bmatrix} \notag \\
    &- \frac{\lambda}{2} \begin{bmatrix}
        I - \eqm I + \xi^2 I^* - \eqm\xi[\varphi]^2 \eqm I[\varphi]^* \\
        0 \\
        0
    \end{bmatrix}, \label{E:PH}
\end{align}
where
\begin{equation}
    \matr F = \begin{bmatrix}
        -D &-1 &\lambda \\
        1 &0 &0 \\
        -\lambda &0 &-R
    \end{bmatrix} .
\end{equation}

\emph{Step 2.} Based on \eqref{E:PH}, we can write
\begin{align}
    \dot{\mathcal H}_{\texttt e}(\vect z_{\texttt e}) &= \left\langle \nabla \mathcal H_{\texttt e}(\vect z_{\texttt e}), \dot{\vect z}_{\texttt e} \right\rangle \notag \\
    &= \left\langle \nabla \mathcal H_{\texttt e}(\vect z_{\texttt e}), \matr F \nabla \mathcal H_{\texttt e}(\vect z_{\texttt e}) \right\rangle + \left\langle g_1, -J \omega^2 \xi + J \eqm\omega^2\eqm\xi[\varphi] \right\rangle \notag \\
    &\quad\, + \left\langle g_1, (j D\eqm\omega+ \lambda \eqm I[\varphi]^* \eqm\xi[\varphi]) (\xi - \eqm\xi[\varphi]) \right\rangle \notag \\
    &\quad\, + \left\langle g_3, V - \eqm V[\varphi] \right\rangle - \frac{\lambda}{2} \left\langle g_1, I-\eqm I[\varphi]\right\rangle \notag \\
    &\quad\, - \frac{\lambda}{2} \left\langle g_1, \xi^2 I^* - \eqm\xi[\varphi]^2 \eqm I[\varphi]^* \right\rangle. \label{E:step1}
\end{align}
We replace the last term in \eqref{E:step1} using the following equality (for brevity, $\eqm\xi = \eqm\xi[\varphi]$ and $\eqm I = \eqm I[\varphi]$),
\begin{align}
    &{-}\left\langle j\omega\xi - j\eqm\omega\eqm\xi, \xi^2 I^*- \eqm\xi^2 \eqm I^* \right\rangle 
    = \left\langle j\omega\xi - j\eqm\omega\eqm\xi, I - \eqm I\right\rangle \notag\\
    &\qquad - \left\langle j \eqm\omega \eqm\xi^* (\xi^2 -\eqm\xi^2), I - \eqm I \right\rangle - \left\langle j\omega\xi - j\eqm\omega\eqm\xi, \eqm I^* (\xi^2 - \eqm\xi^2) \right\rangle. \notag %\label{E:relation}
\end{align}
This gives, after substituting $g_1$, $g_2$, and $g_3$,
\begin{align}
    \dot{\mathcal H}_{\texttt e}(\vect z_{\texttt e}) 
    &= \left\langle \nabla \mathcal H_{\texttt e}(\vect z_{\texttt e}), \matr F \nabla \mathcal H_{\texttt e}(\vect z_{\texttt e}) \right\rangle \notag \\
    &\quad\, + \left\langle j \omega\xi - j \eqm\omega \eqm\xi[\varphi], -J \omega^2 \xi + J \eqm\omega^2\eqm\xi[\varphi] \right\rangle \notag \\
    &\quad\, + \left\langle j \omega \xi - j \eqm\omega \eqm\xi[\varphi], (j D\eqm\omega + \lambda \eqm I[\varphi]^* \eqm\xi[\varphi]) (\xi - \eqm\xi[\varphi]) \right\rangle \notag \\
    &\quad\, +\left\langle I - \eqm I[\varphi], V - \eqm V[\varphi] \right\rangle \notag \\
    &\quad\, - \frac{\lambda}{2} \left\langle j \eqm\omega \eqm\xi[\varphi]^* (\xi^2 - \eqm\xi[\varphi]^2), I - \eqm I[\varphi] \right\rangle \notag \\
    &\quad\, - \frac{\lambda}{2} \left\langle j \omega\xi - j\eqm\omega \eqm\xi[\varphi], \eqm I[\varphi]^* (\xi^2 - \eqm\xi[\varphi]^2) \right\rangle. \notag %\label{E:step2}
\end{align}
We expand the first term, simplify the second term, expand the fifth and sixth terms, and partially cancel the third and the sixth terms to obtain
\begin{align}
    \dot{\mathcal H}_{\texttt e}(\vect z_{\texttt e}) 
    &= -D \|j \omega\xi - j\eqm\omega \eqm\xi[\varphi] \|^2 - R \|I - \eqm I[\varphi]\|^2 \notag \\
    &\quad\, + \langle j \omega\xi - j \eqm\omega\eqm\xi[\varphi], J\eqm\omega \omega(\xi - \eqm\xi[\varphi]) \rangle \notag \\
    &\quad\, + \left\langle j \omega\xi - j \eqm\omega \eqm\xi[\varphi], j D\eqm\omega  (\xi - \eqm\xi[\varphi]) \right\rangle \notag \\
    &\quad\, + \left\langle I - \eqm I[\varphi], V - \eqm V[\varphi] \right\rangle \notag \\
    &\quad\, - \frac{\lambda}{2} \left\langle j \eqm\omega \eqm\xi[\varphi]^* (\xi + \eqm\xi[\varphi]) (\xi - \eqm\xi[\varphi]), I - \eqm I[\varphi] \right\rangle \notag \\
    &\quad\, - \frac{\lambda}{2} \left\langle j \omega\xi - j\eqm\omega \eqm\xi[\varphi], \eqm I[\varphi]^* (\xi - \eqm\xi[\varphi]) (\xi - \eqm\xi[\varphi]) \right\rangle. \label{E:last}
\end{align}
Moreover, from \eqref{E:swing3}, we have
\begin{align}
    &\quad\, \frac{d }{d  t} \left(\frac{\lambda^2}{2L} \|\xi - \eqm\xi[\varphi]\|^2 + \lambda \left\langle I - \eqm I[\varphi], \xi - \eqm\xi[\varphi] \right\rangle \right) \notag \\
    &= \lambda \left\langle j\omega \xi - j \eqm\omega\eqm\xi[\varphi], I - \eqm I[\varphi] \right\rangle - \frac{\lambda R}{L} \left\langle \xi - \eqm\xi[\varphi] , I - \eqm I[\varphi]\right\rangle \notag \\
    &\quad\, + \frac{\lambda}{L} \left\langle \xi - \eqm\xi[\varphi] , V - \eqm V[\varphi]\right\rangle. \label{E:additional}
\end{align}
Combining \eqref{E:last} and \eqref{E:additional} gives
\begin{align}
    \dot{H}_{\texttt e}(\vect x_{\texttt e}, \varphi)
    &= -D \|j \omega\xi - j\eqm\omega \eqm\xi[\varphi] \|^2 - R \|I - \eqm I[\varphi]\|^2 \notag \\
    &\quad\, + \langle j \omega\xi - j \eqm\omega\eqm\xi[\varphi], J\eqm\omega \omega(\xi - \eqm\xi[\varphi]) \rangle \notag \\
    &\quad\, + \left\langle j \omega\xi - j \eqm\omega \eqm\xi[\varphi], j D\eqm\omega  (\xi - \eqm\xi[\varphi]) \right\rangle \notag \\
    &\quad\, + \left\langle I - \eqm I[\varphi], V - \eqm V[\varphi] \right\rangle \notag \\
    &\quad\, - \frac{\lambda}{2} \left\langle j \eqm\omega \eqm\xi[\varphi]^* (\xi - \eqm\xi[\varphi]) (\xi - \eqm\xi[\varphi]), I - \eqm I[\varphi] \right\rangle \notag \\
    &\quad\, - \frac{\lambda}{2} \left\langle j \omega\xi - j\eqm\omega \eqm\xi[\varphi], \eqm I[\varphi]^* (\xi - \eqm\xi[\varphi]) (\xi - \eqm\xi[\varphi]) \right\rangle \notag \\
    &\quad\, + \lambda \left\langle j (\omega - \eqm\omega) \xi, I - \eqm I[\varphi] \right\rangle \notag \\
    &\quad\, - \frac{\lambda R}{L} \left\langle \xi - \eqm\xi[\varphi], I - \eqm I[\varphi] \right\rangle \notag \\
    &\quad\, + \frac{\lambda}{L} \left\langle \xi - \eqm\xi[\varphi] , V - \eqm V[\varphi]\right\rangle, \notag
\end{align}
which rearranges to \eqref{E:eb_raw_raw}.
\end{proof}

\subsection{Proof of Lemma \ref{lem:phase-dynamics}} \label{proof_phase_dynamics}

\begin{proof}
All barred (steady-state) quantities are evaluated at the reference phase
$\hat\varphi$.

\emph{Step 1 (Extracting the quadratic phase contribution).}
Using \eqref{E:aggregate} and the individual edge storage functions, we obtain
\begin{align*}
    &\quad\,\, \frac{\partial \widetilde H}{\partial \hat \varphi} (\vect x, \hat\varphi) \\
    &= -\sum_{\texttt e=1}^{\texttt g} \bigg[\sigma_{\texttt e} J_{\texttt e} \nu_{\texttt e} + \left(\sigma_{\texttt e} D_{\texttt e} + 2\mathcal K_{\texttt e} + \frac{\lambda_{\texttt e}^2}{L_{\texttt e}} \right)\hat\delta_{\texttt e} \\
    &\qquad\qquad\qquad \qquad\,  + \big\langle \hat\zeta_{\texttt e}, j \left(L_{\texttt e} \eqm I_{\texttt e} + \lambda_{\texttt e} \eqm\xi_{\texttt e} \right) \big\rangle \bigg] \\
    &\quad\, - \sum_{\texttt e=\texttt g+1}^{\texttt g+\texttt T} L_{\texttt e} \big\langle \hat\zeta_{\texttt e}, j \eqm I_{\texttt e} \big\rangle - \sum_{\texttt e=\texttt g+\texttt T+1}^{2\texttt g+\texttt T+\ell} C_{\texttt e}\left\langle \hat\eta_{\texttt e}, j \eqm V_{\texttt e} \right\rangle + O(\|\hat{\vect z}\|^2).
\end{align*}
Let
\[
\frac{\partial\widetilde H}{\partial\hat\varphi}
(\vect x, \hat\varphi)
=\mathcal Z(\hat{\vect z})+O(\|\hat{\vect z}\|^2),
\]
where $\mathcal Z(\hat{\vect z})$ is the linear part of $\frac{\partial\widetilde H}{\partial\hat\varphi}(\vect x, \hat\varphi)$.
We have
\[
\left[
\frac{\partial\widetilde H}{\partial\hat\varphi}
(\vect x, \hat\varphi)\right]^2
=\mathcal Z(\hat{\vect z})^2+O(\|\hat{\vect z}\|^3).
\]
Hence the negative of the quadratic part of \eqref{E:last_balance} is
\[
\mathcal P(c)
\coloneqq
\widetilde{\mathcal Q}(\hat{\vect z})
-\sum_{\texttt i=1}^{\texttt g}
\frac{\lambda_{\texttt i}}{L_{\texttt i}}\hat\delta_1
\big\langle\hat\eta_{\texttt i+\texttt g+\texttt T},j\bar\xi_{\texttt i} \big\rangle
+c\mathcal Z(\hat{\vect z})^2.
\]
It remains to determine when $\mathcal P(c)\succ0$.

\emph{Step 2 (Eliminate the non-angle SG variables).}
Define
\begin{align*}
&\vect z_{\texttt e}\coloneqq
\begin{bmatrix}
\nu_{\texttt e} &\hat\delta_{\texttt e} &
-\big\langle\hat\zeta_{\texttt e}, j\bar\xi_{\texttt e}\big\rangle
\end{bmatrix}^\tran, \quad \mathcal I_{\texttt e}\coloneqq\big\langle\hat\zeta_{\texttt e},\bar\xi_{\texttt e}\big\rangle.
\end{align*}
The generator dissipation and its contribution to $\mathcal Z$ are
\[
\vect z_{\texttt e}^\tran \matr M_{\texttt e}(\sigma_{\texttt e})\vect z_{\texttt e}+R_{\texttt e}\mathcal I_{\texttt e}^2,
\quad
\vect b_{\texttt e}^\tran \vect z_{\texttt e}+
L_{\texttt e}\langle\bar I_{\texttt e},j\bar\xi_{\texttt e}\rangle\mathcal I_{\texttt e}.
\]
Decompose
\begin{equation*}
    \vect z_{\texttt e} = \frac{\hat\delta_{\texttt e}}{d_{\texttt e}} \matr M_{\texttt e}(\sigma_{\texttt e})^{-1} \vect e_2 + \vect v_{\texttt e}
\end{equation*}
such that $\vect e_2^\tran \vect v_{\texttt e} = 0$.
Since $d_{\texttt e} = \vect e_2^\tran \matr M_{\texttt e}(\sigma_{\texttt e})^{-1} \vect e_2$, we have
\begin{equation*}
    \vect z_{\texttt e}^\tran \matr M_{\texttt e}(\sigma_{\texttt e}) \vect z_{\texttt e} = \frac{\hat\delta_{\texttt e}^2}{d_{\texttt e}} + \vect v_{\texttt e}^\tran \matr M_{\texttt e}(\sigma_{\texttt e}) \vect v_{\texttt e}.
\end{equation*}
Furthermore, 
\begin{equation*}
    \vect b_{\texttt e}^\tran \vect z_{\texttt e} = \frac{\vect e_2^\tran \matr M_{\texttt e}(\sigma_{\texttt e})^{-1} \vect b_{\texttt e}}{d_{\texttt e}} \hat\delta_{\texttt e} + \vect b_{\texttt e}^\tran \vect v_{\texttt e} = r_{\texttt e}\hat\delta_{\texttt e} + \vect b_{\texttt e}^\tran \vect v_{\texttt e}.
\end{equation*}
Thus, eliminating $\vect v_{\texttt e}$ leaves the angle-dissipation coefficient
\begin{equation*}\alpha_{\texttt e}\coloneqq 1/d_{\texttt e}
\end{equation*}
and changes the coefficient of $\hat\delta_{\texttt e}$ in $\mathcal Z$ to
\begin{equation*}
r_{\texttt e}\coloneqq n_{\texttt e}/d_{\texttt e}.
\end{equation*}
The squared dual norm of the remaining functional $\vect v_{\texttt e} \mapsto \vect b_{\texttt e}^\tran \vect v_{\texttt e}$, on the subspace $\vect e_2^\tran \vect v_{\texttt e} = 0$ is
\begin{align*}
    &\quad\,\, \vect b_{\texttt e}^\tran \left(\matr M_{\texttt e}(\sigma_{\texttt e})^{-1} - \frac{\matr M_{\texttt e}(\sigma_{\texttt e})^{-1} \vect e_2 \vect e_2^\tran \matr M_{\texttt e}(\sigma_{\texttt e})^{-1}}{d_{\texttt e}} \right) \vect b_{\texttt e} \\
    &= \vect b_{\texttt e}^\tran \matr M_{\texttt e}(\sigma_{\texttt e})^{-1} \vect b_{\texttt e} - \frac{n_{\texttt e}^2}{d_{\texttt e}}.
\end{align*}
The independent coordinate $\mathcal I_{\texttt e}$ has dissipation $R_{\texttt e}\mathcal I_{\texttt e}^2$ and coefficient $L_{\texttt e} \left\langle \eqm I_{\texttt e}, j \eqm\xi_{\texttt e} \right\rangle$ in $\mathcal Z$; eliminating it therefore contributes
\begin{equation*}
    \frac{L_{\texttt e}^2 \left\langle \eqm I_{\texttt e}, j \eqm\xi_{\texttt e} \right\rangle^2}{R_{\texttt e}}
\end{equation*}
to the squared dual norm. Similar contributions are obtained for the $R$--$L$ line currents.

\emph{Step 3 (Eliminate the shunt capacitor voltages).}
Define the shifted shunt-voltage coordinates
\begin{equation*}
    \tilde\eta_{\texttt i+\texttt g+\texttt T} \coloneqq \hat \eta_{\texttt i+\texttt g+\texttt T} - \frac{\lambda_{\texttt i}}{2L_{\texttt i} G_{\texttt i+\texttt g+\texttt T}} j \eqm\xi_{\texttt i} \hat\delta_1.
\end{equation*}
In terms of this coordinate, the dissipation becomes
\begin{align*}
&
\quad\,\, G_{\texttt i+\texttt g+\texttt T}\|\hat\eta_{\texttt i+\texttt g+\texttt T}\|^2
-\frac{\lambda_{\texttt i}}{L_{\texttt i}}\hat\delta_1
\langle\hat\eta_{\texttt i+\texttt g+\texttt T},j\bar\xi_{\texttt i}\rangle
\\
&=
G_{\texttt i+\texttt g+\texttt T}
\left\|
\tilde\eta_{\texttt i+\texttt g+\texttt T}
\right\|^2
-
\frac{(\lambda_{\texttt i}/L_{\texttt i})^2}
{4G_{\texttt i+\texttt g+\texttt T}}\hat\delta_1^2.
\end{align*}
Thus define
\[
\kappa\coloneqq
\sum_{\texttt i=1}^{\texttt g}
\frac{(\lambda_{\texttt i}/L_{\texttt i})^2}{4 G_{\texttt i+\texttt g+\texttt T}}.
\]
The same substitution in the shunt contribution to $\mathcal Z$ produces the additional angle term
\begin{equation*}
    -\frac{\lambda_{\texttt i} C_{\texttt i+\texttt +\texttt T}}{2L_{\texttt i} G_{\texttt i+\texttt g+\texttt T}} \left\langle \eqm\xi_{\texttt i}, \eqm V_{\texttt i+\texttt g+\texttt T} \right\rangle \hat\delta_1.
\end{equation*}
Summing over $\texttt i$ therefore changes the coefficient of $\hat\delta_1$ from $r_1$ to $r_1 - q$.

Collecting all remaining non-angle variables into $\vect w$ gives
\begin{align*}
\mathcal P(c)
&=
\sum_{\texttt e=1}^{\texttt g}\alpha_{\texttt e}\hat\delta_{\texttt e}^2
-\kappa\hat\delta_1^2+ \vect w^\tran \matr D \vect w
\\
&\quad\, +
c\left[
(r_1-q)\hat\delta_1
+\sum_{\texttt e=2}^{\texttt g} r_{\texttt e}\hat\delta_{\texttt e}
+\vect a^\tran \vect w
\right]^2,
\end{align*}
where $\matr D\succ0$ and
\[
\vect a^\tran \vect D^{-1} \vect a 
=
\sum_{\texttt e=1}^{\texttt g}
\left(
\vect b_{\texttt e}^\tran \matr M_{\texttt e}(\sigma_{\texttt e})^{-1} \vect b_{\texttt e}-\frac{n_{\texttt e}^2}{d_{\texttt e}}
\right)
+\tau.
\]

\emph{Step 4 (Eliminate $\vect w$ and
$\hat\delta_2,\ldots,\hat\delta_\textup{\texttt g}$).}
Note that, for $\matr D\succ0$,
\begin{equation} \label{E:woodbury}
\min_{\vect w}
\left\{
\vect w^\tran \matr D\vect w+c(\rho+\vect a^\tran \vect w)^2
\right\}
=
\frac{c}{1+c\,\vect a^\tran \matr D^{-1}\vect a}\rho^2,
\end{equation}
which is a consequence of the Woodbury formula~\cite{bernstein2009matrix}.

Define
\begin{equation*}
    \mathcal B(c) \coloneqq \frac{c}{1 + c \mathcal A},\quad \mathcal A \coloneqq \vect a^\tran \matr D^{-1} \vect a.
\end{equation*}
After eliminating $\vect w$, the remaining angle quadratic form is
\begin{equation*}
    (\alpha_1 - \kappa) \hat\delta_1^2 + \sum_{\texttt e=2}^{\texttt g} \alpha_{\texttt e} \hat\delta_{\texttt e}^2 + \mathcal B(c) \left[(r_1 - q) \hat\delta_1 + \sum_{\texttt e=2}^{\texttt g} r_{\texttt e} \hat\delta_{\texttt e} \right]^2.
\end{equation*}
A second application of \eqref{E:woodbury} yields
\begin{align*}
    &\quad\,\, \min_{\hat\delta_2,\ldots,\hat\delta_{\texttt g}} \left\{ \sum_{\texttt e=2}^{\texttt g} \alpha_{\texttt e} \hat\delta_{\texttt e}^2 + \mathcal B(c) \left[(r_1 - q) \hat\delta_1 + \sum_{\texttt e=2}^{\texttt g} r_{\texttt e} \hat\delta_{\texttt e} \right]^2\right\} \\
    &= \frac{\mathcal B(c)}{1 + \mathcal B(c) \sum_{\texttt e=2}^{\texttt g} r_{\texttt e}^2/\alpha_{\texttt e}} (r_1 - q)^2 \hat\delta_1^2.
\end{align*}
Since
\begin{equation*}
    r_{\texttt e}^2/\alpha_{\texttt e} = n_{\texttt e}^2/d_{\texttt e},
\end{equation*}
and
\begin{equation*}
    \mathcal A = \sum_{\texttt e=1}^{\texttt g} \left(\vect b_{\texttt e}^\tran \matr M_{\texttt e}(\sigma_{\texttt e})^{-1} \vect b_{\texttt e} - \frac{n_{\texttt e}^2}{d_{\texttt e}} \right) + \tau,
\end{equation*}
we have
\begin{equation*}
    \mathcal A + \sum_{\texttt e=2}^{\texttt g} \frac{n_{\texttt e}^2}{d_{\texttt e}} = \sum_{\texttt e=1}^{\texttt g} \vect b_{\texttt e}^\tran \matr M_{\texttt e}(\sigma_{\texttt e})^{-1} \vect b_{\texttt e} - \frac{n_1^2}{d_1} + \tau \eqqcolon \Theta.
\end{equation*}
Consequently, the final scalar Schur complement is
\begin{equation*}
    \mathcal F(c)
\coloneqq
\alpha_1-\kappa
+\frac{c}{1+c\Theta}(r_1-q)^2
\end{equation*}

\emph{Step 5 (Optimize over $c$).}
The projected $\matr M_1(\sigma_1)^{-1}$ term is nonnegative by
Cauchy--Schwarz, and the steady-state stator equations imply
$\tau>0$. Hence $\Theta>0$. Therefore $\mathcal F(c)$ is nondecreasing
and
\begin{align*}
\lim_{c\to\infty}\mathcal F(c)
&=
\alpha_1
-
\kappa
+
\frac{
\left(
r_1
-q
\right)^2
}{\Theta}.
\end{align*}
This is precisely the left-hand side of
\eqref{eq:phase-dynamics-condition}. Thus a finite $c>0$ satisfying
$\mathcal F(c)>0$ exists if and only if \eqref{eq:phase-dynamics-condition} holds.
Equivalently, the quadratic part of \eqref{E:last_balance} is
negative definite.
\end{proof}

%\balance

\bibliographystyle{IEEEtran} 
\bibliography{references}

\end{document}